\documentclass[a4paper,10pt]{article}
\usepackage{geometry}
\usepackage{authblk}
\usepackage[parfill]{parskip}
\usepackage[round,authoryear,semicolon]{natbib}
\usepackage[bookmarks=false,colorlinks]{hyperref}
\usepackage{mdl}
\usepackage{autonum}
\title{Minimax Optimal Estimators for Additive Scalar Functionals of Discrete Distributions}
\author[1,2]{Kazuto Fukuchi~\thanks{kazuto@mdl.cs.tsukuba.ac.jp}}
\author[1,3,4]{Jun Sakuma~\thanks{jun@cs.tsukuba.ac.jp}}
\affil[1]{Department of Computer Science, Graduate School of System and Information Engineering, University of Tsukuba}
\affil[2]{Japan Society for Promotion of Science}
\affil[3]{JST CREST}
\affil[4]{RIKEN Center for Advanced Intelligence Project}
\begin{document}
\maketitle
\begin{abstract}
\noindent In this paper, we consider estimators for an {\em additive functional} of $\phi$, which is defined as $\theta(P;\phi)=\sum_{i=1}^k\phi(p_i)$, from $n$ i.i.d. random samples drawn from a discrete distribution $P=(p_1,...,p_k)$ with alphabet size $k$.  We propose a minimax optimal estimator for the estimation problem of the additive functional. We reveal that the minimax optimal rate is characterized by the {\em divergence speed} of the fourth derivative of $\phi$ if the divergence speed is high. As a result, we show there is no consistent estimator if the divergence speed of the fourth derivative of $\phi$ is larger than $p^{-4}$. Furthermore, if the divergence speed of the fourth derivative of $\phi$ is $p^{4-\alpha}$ for $\alpha \in (0,1)$, the minimax optimal rate is obtained within a universal multiplicative constant as $\frac{k^2}{(n\ln n)^{2\alpha}} + \frac{k^{2-2\alpha}}{n}$.
\end{abstract}

\section{Introduction}

Let $P$ be a probability measure with alphabet size $k$, and $X$ be a discrete random variable drawn from $P$. Without loss of generality, we can assume that the domain of $P$ is $[k]$, where we denote $[m]=\cbrace{1,...,m}$ for a positive integer $m$. We use a vector representation of $P$; $P=(p_1,...,p_k)$ where $p_i=P\cbrace{X=i}$. Let $\phi$ be a mapping from $[0,1]$ to $\RealSet^+$. Given a set of i.i.d. samples $S_n=\cbrace{X_1,...,X_n}$ from $P$, we deal with the problem of estimating an {\em additive functional} of $\phi$. The additive functional $\theta$ of $\phi$ is defined as
\begin{align}
 \theta(P;\phi) = \sum_{i=1}^k \phi(p_i).
\end{align}
We simplify this notation to $\theta(P;\phi)=\theta(P)$. Most entropy-like criteria can be formed in terms of $\theta$. For instance, when $\phi(p)=-p\ln p$, $\theta$ is Shannon entropy. For a positive real $\alpha$, letting $\phi(p)=p^\alpha$, $\ln(\theta(P))/(1-\alpha)$ becomes R\'enyi entropy. More generally, letting $\phi = f$ where $f$ is a concave function, $\theta$ becomes $f$-entropies~\citep{akaike1998information}.

Techniques for the estimation of the entropy-like criteria have been considered in various fields, including physics~\citep{lake2011accurate}, neuroscience~\citep{nemenman2004entropy}, and security~\citep{gu2005detecting}. In machine learning, methods that involve entropy estimation were introduced for decision-trees~\citep{quinlan1986induction}, feature selection~\citep{peng2005feature}, and clustering~\citep{dhillon2003information}. For example, the decision-tree learning algorithms, i.e., ID3, C4.5, and C5.0 construct a decision tree in which the criteria for the tree splitting are defined based on Shannon entropy~\citep{quinlan1986induction}. Similarly, information theoretic feature selection algorithms evaluate the relevance between the features and the target using the entropy~\citep{peng2005feature}.

The goal of this study is to derive the minimax optimal estimator of $\theta$ given a function $\phi$. For the precise definition of the minimax optimality, we introduce the minimax risk. A sufficient statistic of $P$ is a histogram $N=\paren{N_1,...,N_k}$, where $N_j = \sum_{i=1}^n \ind{\cbrace{X_i = j}}$ and $N \sim \Mul(n, P)$. The estimator of $\theta$ is defined as a function $\hat\theta:[n]^k\to\RealSet$. Then, the quadratic minimax risk is defined as
\begin{align}
 R^*(n,k;\phi) = \inf_{\hat\theta}\sup_{P \in \dom{M}_k} \Mean\bracket*{\paren*{\hat\theta(N) - \theta(P)}^2},
\end{align}
where $\dom{M}_k$ is the set of all probability measures on $[k]$, and the infimum is taken over all estimators $\hat\theta$. With this definition of the minimax risk, an estimator $\hat\theta$ is minimax \mbox{(rate-)optimal} if there exists a positive constant $C$ such that
\begin{align}
 \sup_{P \in \dom{M}_k} \Mean\bracket*{\paren*{\hat\theta(N) - \theta(P)}^2} \le C R^*(n,k;\phi).
\end{align}

A natural estimator of $\theta$ is the plugin or the maximum likelihood estimator, in which the estimated value is obtained by substituting the empirical mean of the probabilities $P$ into $\theta$.  However, the estimator has a large bias for large $k$. Indeed, the plugin estimators for $\phi(p)=-p\ln p$ and $\phi(p)=p^\alpha$ have been shown to be suboptimal in the large-$k$ regime in recent studies~\citep{wu2016minimax,jiao2015minimax,DBLP:conf/soda/AcharyaOST15}.

Recent studies investigated the minimax optimal estimators for $\phi(p)=-p\ln p$ and $\phi(p)=p^\alpha$ in the large-$k$ regime~\citep{wu2016minimax,jiao2015minimax,DBLP:conf/soda/AcharyaOST15}. However, the results of these studies were only derived for these $\phi$. \citet{jiao2015minimax} suggested that the estimator is easily extendable to the general additive functional, although they did not prove the minimax optimality.

In this paper, we propose a minimax optimal estimator for the estimation problem of the additive functional $\theta$ for general $\phi$ under certain conditions on the smoothness. Our estimator achieves the minimax optimal rate even in the large-$k$ regime for $\phi \in C^4[0,1]$ such that $\abs*{\phi^{(4)}(p)}$ is finite for $p \in (0,1]$, where $C^4[0,1]$ denotes a class of four times differentiable functions from $[0,1]$ to $\RealSet$. For such $\phi$, we reveal a property of $\phi$ which can substantially influence the minimax optimal rate.

{\bfseries Related work.}
The simplest way to estimate $\theta$ is to use the so-called plugin estimator or the maximum likelihood estimator, in which the empirical probabilities are substituted into $\theta$ as $P$. Letting $\tilde{P} = (\hat{p}_1,...,\hat{p}_k)$ and $\hat{p}_i = N_i/n$, the plugin estimator is defined as
\begin{align}
  \theta_{\rm plugin}(N) = \theta(\tilde{P}).
\end{align}
The plugin estimator is asymptotically consistent under weak assumptions for fixed $k$~\citep{RSA:RSA10019}. However, this is not true for the large-$k$ regime. Indeed, \citet{jiao2015minimax} and \citet{wu2016minimax} derived a lower bound for the quadratic risk for the plugin estimator of $\phi(p)=p\ln(1/p)$ and $\phi(p)=p^\alpha$. In the case of Shannon entropy, the lower bound is given as
\begin{align}
  \sup_{P \in \dom{M}_k}\Mean\bracket*{\paren*{\theta_{\rm plugin}(N) - \theta(P)}^2} \ge C\paren*{\frac{k^2}{n^2} + \frac{\ln^2 k}{n}},
\end{align}
where $C$ denotes a universal constant. The first term $k^2/n^2$ comes from the bias and it indicates that if $k$ grows linearly with respect to $n$, the plugin estimator becomes inconsistent. This means the plugin estimator is suboptimal in the large-$k$ regime. Bias-correction methods, such as \citep{miller1955nbi,grassberger1988finite,10.2307/1936227}, can be applied to the plugin estimator of $\phi(p)=-p\ln p$ to reduce the bias whereas these bias-corrected estimators are still suboptimal. The estimators based on Bayesian approaches in \citep{schurmann1996entropy,6620615,holste1998bayes} are also suboptimal~\citep{DBLP:journals/corr/HanJW15a}.

Many researchers have studied estimators that can consistently estimate the additive functional with sublinear samples with respect to the alphabet size $k$ to derive the optimal estimator in the large-$k$ regime. The existence of consistent estimators even with sublinear samples were first revealed in \citet{paninski2004estimating}, but an explicit estimator was not provided. \citet{DBLP:conf/stoc/ValiantV11} introduced an estimator based on linear programming that consistently estimates $\phi(p)=-p\ln p$ with sublinear samples. However, the estimator of \citep{DBLP:conf/stoc/ValiantV11} has not been shown to achieve the minimax rate even in a more detailed analysis in \citep{DBLP:conf/focs/ValiantV11}. Recently, \citet{DBLP:conf/soda/AcharyaOST15} showed that the bias-corrected estimator of R\'enyi entropy achieves the minimax optimal rate in regard to the sample complexity if $\alpha > 1$ and $\alpha \in \NaturalSet$, but they did not show the minimax optimality for other $\alpha$. \citet{jiao2015minimax} introduced a minimax optimal estimator for $\phi(p)=-p\ln p$ for any $\alpha \in (0,3/2)$ in the large-$k$ regime. \citet{2015arXiv150401227W} derived a minimax optimal estimator for $\phi(p)=\ind{p > 0}$. For $\phi(p)=-p\ln p$, \citet{jiao2015minimax,wu2016minimax} independently introduced the minimax optimal estimators in the large-$k$ regime. In the case of Shannon entropy, the optimal rate was obtained as
\begin{align}
  \frac{k^2}{(n\ln n)^2} + \frac{\ln^2 k}{n}.
\end{align}
The first term indicates that the introduced estimator can consistently estimate Shannon entropy if $n \ge C k/\ln k$.

The estimators introduced by \citet{wu2016minimax,jiao2015minimax,DBLP:conf/soda/AcharyaOST15} are composed of two estimators: the bias-corrected plugin estimator and the best polynomial estimator. The bias-corrected plugin estimator is composed of the sum of the plugin estimator and a bias-correction term which offsets the second-order approximation of the bias as in \citep{miller1955nbi}. The best polynomial estimator is an unbiased estimator of the polynomial that best approximates $\phi$ in terms of the uniform error. Specifically, the best approximation for the polynomial of $\phi$ in an interval $I \subseteq [0,1]$ is the polynomial $g$ that minimizes $\sup_{x \in I}\abs*{\phi(x) - g(x)}$. \citet{jiao2015minimax} suggested that this estimator can be extended for the general additive functional $\theta$. However, the minimax optimality of the estimator was only proved for specific cases of $\phi$, including $\phi(p)=-p\ln p$ and $\phi(p)=p^\alpha$. Thus, to prove the minimax optimality for other $\phi$, we need to individually analyze the minimax optimality for specific $\phi$. Here, we aim to clarify which property of $\phi$ substantially influences the minimax optimal rate when estimating the additive functional.

Besides, the optimal estimators for divergences with large alphabet size have been investigated in \citep{7541473,7840425,7541399}. The estimation problems of divergences are much complicated than the additive function, while the similar techniques were applied to derive the minimax optimality.

\begin{figure}[t]
  \centering
  \includegraphics[width=.9\textwidth,draft=false]{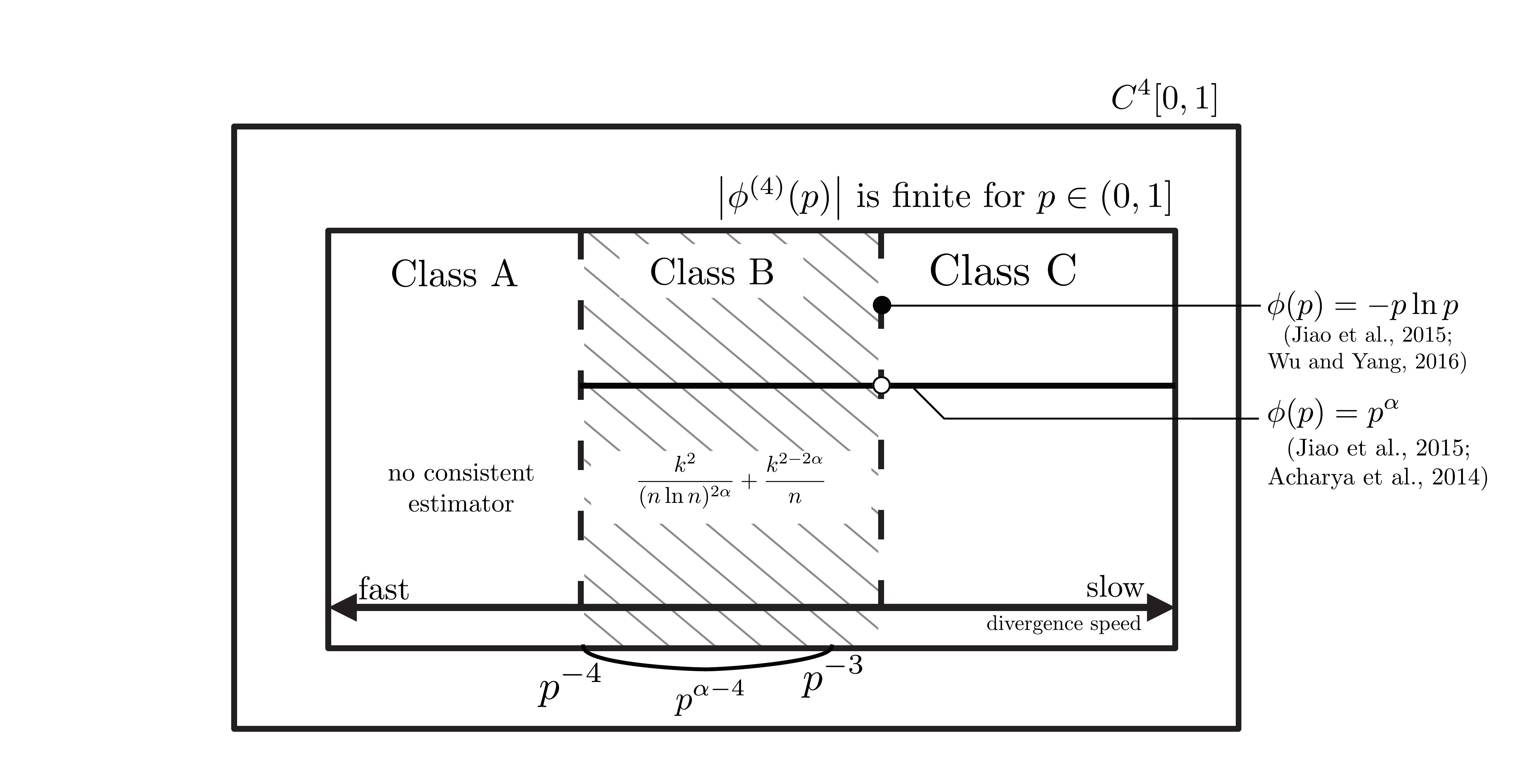}
  \caption{Relationship between the divergence speed of the fourth derivative of $\phi$ and the minimax optimality of the estimation problem of $\theta(P;\phi)$.}\label{fig:overview}
\end{figure}

{\bfseries Our contributions.}
In this paper, we propose the minimax optimal estimator for $\theta(P;\phi)$. We reveal that the {\em divergence speed} of the fourth derivative of $\phi$ plays an important role in characterizing the minimax optimal rate. Informally, for $\beta > 0$, the meaning of ``the divergence speed of a function $f(p)$ is $p^{-\beta}$'' is that $\abs*{f(p)}$ goes to infinity at the same speed as $p^{-\beta}$ when $p$ approaches $0$. When the divergence speed of the fourth derivative of $\phi(p)$ is $p^{-\beta}$, the fourth derivative of $\phi$ diverges faster as $\beta$ increases.

Our results are summarized in \cref{fig:overview}. \cref{fig:overview} illustrates the relationship between the divergence speed of the fourth derivative of $\phi$ and the minimax optimality of the estimation problem of $\theta(P;\phi)$. In \cref{fig:overview}, the outermost rectangle represents the space of the four times continuous differentiable functions $C^4[0,1]$. The innermost rectangle denotes the subset class of $C^4[0,1]$ such that the absolute value of its fourth derivative $\abs*{\phi^{(4)}(p)}$ is finite for any $p \in (0,1]$. In this subclass of $\phi$, the horizontal direction represents the divergence speed of the fourth derivative of $\phi$, in which a faster $\phi$ is on the left-hand side and a slower $\phi$ is on the right-hand side. The $\phi$ with an explicit form and divergence speed is denoted by a point in the rectangle. For example, the black circle denotes $\phi(p)=-p\ln p$ where the divergence speed of the fourth derivative of this $\phi$ is $p^{-3}$. Class B denotes a set of any function $\phi$ such that the divergence speed of the fourth derivative is $p^{\alpha-4}$ where $\alpha\in(0,1)$. As already discussed, existing methods have achieved minimax optimality in the large-$k$ regime for specific $\phi$, including $\phi(p)=-p\ln p$~(black circle in \cref{fig:overview}) and $\phi(p)=p^\alpha$~(middle line in \cref{fig:overview} where the white circle denotes that there is no $\alpha > 0$ such that the divergence speed is $p^{-3}$).

We investigate the minimax optimality of the estimation problem of $\theta$ for $\phi$ in Class A and Class B. Class A is a class of $\phi$ such that the divergence speed of the fourth derivative is faster than $p^{-4}$. Class B is a class of $\phi$ such that the divergence speed of the fourth derivative is $p^{\alpha-4}$ where $\alpha \in (0,1)$. In Class A, we show that we cannot construct a consistent estimator of $\theta$ for any $\phi$ in Class A~(the leftmost hatched area in \cref{fig:overview}, \cref{prop:no-estimator}). In other words, the minimax optimal rate is larger than constant order if the divergence speed of the fourth derivative is faster than $p^{-4}$. Thus, there is no need to derive the minimax optimal estimator in Class A.

Also, we derive the minimax optimal estimator for any $\phi$ in Class B~(the middle hatched area in \cref{fig:overview}, \cref{thm:minimax-rate}). For example, $\phi(p)=p^\alpha$~(R\'eyni entropy case), $\phi(p)=\cos(cp)p^\alpha$, and $\phi(p)=e^{cp}p^\alpha$ for $\alpha \in (0,1)$ include the coverage of our estimator, where $c$ is a universal constant. Intuitively, since the large derivative makes the estimation problem $\theta$ more difficult, the minimax rate decreases if the derivative of $\phi$ diverges faster. Our minimax optimal rate reflects this behavior. For $\phi$ in Class B, the minimax optimal rate is obtained as
\begin{align}
  \frac{k^2}{(n\ln n)^{2\alpha}} + \frac{k^{2-2\alpha}}{n},
\end{align}
where $k \gtrsim \ln^{\frac{4}{3}}n$ if $\alpha \in (0,1/2]$. We can clearly see that this rate decreases for larger $\alpha$, i.e., a slower divergence speed.

Currently, the minimax optimality of $\phi$ in Class C is an open problem. However, we provide a notable discussion in \cref{sec:main-results}.

\section{Preliminaries}

\noindent{\bfseries Notations.}
We now introduce some additional notations. For any positive real sequences $\cbrace{a_n}$ and $\cbrace{b_n}$, $a_n \gtrsim b_n$ denotes that there exists a positive constant $c$ such that $a_n \ge c b_n$. Similarly, $a_n \lesssim b_n$ denotes that there exists a positive constant $c$ such that $a_n \le c b_n$. Furthermore, $a_n \asymp b_n$ implies $a_n \gtrsim b_n$ and $a_n \lesssim b_n$. For an event $\event$, we denote its complement by $\event^c$. For two real numbers $a$ and $b$, $a \lor b = \max\cbrace{a, b}$ and $a \land b = \min\cbrace{a, b}$. For a function $\phi:\RealSet\to\RealSet$, we denote its $i$-th derivative as $\phi^{(i)}$.

\noindent{\bfseries Poisson sampling.}
We employ the Poisson sampling technique to derive upper and lower bounds for the minimax risk. The Poisson sampling technique models the samples as independent Poisson distributions, while the original samples follow a multinomial distribution. Specifically, the sufficient statistic for $P$ in the Poisson sampling is a histogram $\tilde{N} = \paren{\tilde{N}_i,...,\tilde{N}_k}$, where $\tilde{N}_1,...,\tilde{N}_k$ are independent random variables such that $\tilde{N}_i \sim \Poi(np_i)$. The minimax risk for Poisson sampling is defined as follows:
\begin{align}
 \tilde{R}^*(n,k;\phi) = \inf_{\cbrace{\hat\theta}}\sup_{P \in \dom{M}_k} \Mean\bracket*{\paren*{\hat\theta(\tilde{N}) - \theta(P)}^2}.
\end{align}
The minimax risk of Poisson sampling well approximates that of the multinomial distribution. Indeed, \citet{jiao2015minimax} presented the following lemma.
\begin{lemma}[\citet{jiao2015minimax}]\label{lem:well-approx-poisson}
  The minimax risk under the Poisson model and the multinomial model are related via the following inequalities:
  \begin{align}
    \tilde{R}^*(2n,k;\phi) - \sup_{P \in \dom{M}_k}\abs*{\theta(P)}e^{-n/4} \le R^*(n,k;\phi) \le 2\tilde{R}^*(n/2,k;\phi).
  \end{align}
\end{lemma}
\cref{lem:well-approx-poisson} states $R^*(n,k;\phi) \asymp \tilde{R}^*(n,k;\phi)$, and thus we can derive the minimax rate of the multinomial distribution from that of the Poisson sampling.

{\bfseries Best polynomial approximation.}
\citet{DBLP:conf/soda/AcharyaOST15,wu2016minimax,jiao2015minimax} presented a technique of the best polynomial approximation for deriving the minimax optimal estimators and their lower bounds for the risk. Let $\dom{P}_L$ be the set of polynomials of degree $L$. Given a function $\phi$, a polynomial $p$, and an interval $I \subseteq [0,1]$, the uniform error between $\phi$ and $p$ on $I$ is defined as
\begin{align}
  \sup_{x \in I}\abs*{\phi(x) - p(x)}. \label{eq:uniform-error}
\end{align}
The best polynomial approximation of $\phi$ by a degree-$L$ polynomial with a uniform error is achieved by the polynomial $p \in \dom{P}_L$ that minimizes \cref{eq:uniform-error}. The error of the best polynomial approximation is defined as
\begin{align}
  E_L\paren*{\phi, I} = \inf_{p \in \dom{P}_L}\sup_{x \in I}\abs*{\phi(x) - p(x)}.
\end{align}
The error rate with respect to the degree $L$ has been studied since the 1960s~\citep{timan1965theory,petrushev2011rational,ditzian2012moduli,achieser2013theory}. The polynomial that achieves the best polynomial approximation can be obtained, for instance, by the Remez algorithm~\citep{remez1934determination} if $I$ is bounded.

\section{Main results}\label{sec:main-results}

Suppose $\phi$ is four times continuously differentiable on $(0,1]$\footnote{We say that a function $\phi:[0,1]\to\RealSet_+$ is differentiable at $1$ if $\lim_{h \to -0}\frac{\phi(1+h)-\phi(1)}{h}$ exists.}. We reveal that the {\em divergence speed} of the fourth derivative of $\phi$ plays an important role for the minimax optimality of the estimation problem of the additive functional. Formally, the divergence speed is defined as follows.
\begin{definition}[divergence speed]\label{def:div-speed}
  For an integer $m \ge 1$, let $\phi$ be an $m$ times continuously differentiable function on $(0,1]$. For $\beta > 0$, the divergence speed of the $m$th derivative of $\phi$ is $p^{-\beta}$ if there exist finite constants $W > 0$, $c_m$, and $c'_m$ such that for all $p \in (0,1]$
  \begin{align}
    \abs*{\phi^{(m)}(p)} \le \beta_{m-1} Wp^{-\beta} + c_m, \textand \abs*{\phi^{(m)}(p)} \ge \beta_{m-1} Wp^{-\beta} + c'_m,
  \end{align}
  where $\beta_m = \prod_{i=1}^m (i-m+\beta)$.
\end{definition}
A larger $\beta$ implies faster divergence. We analyze the minimax optimality for two cases: the divergence speed of the fourth derivative of $\phi$ is i) larger than $p^{-4}$~(Class A), and ii) $p^{\alpha-4}$~(Class B), for $\alpha \in (0,1)$.

{\bfseries Minimax optimality for Class A.}
We now demonstrate that we cannot construct a consistent estimator for any $n$ and $k \ge 3$ if the divergence speed of $\phi$ is larger than $p^{-4}$.
\begin{proposition}\label{prop:no-estimator}
  Let $\phi$ be a continuously differentiable function on $(0,1]$. If there exists finite constants $W > 0$ and $c'_1$ such that for $p \in (0,1]$
  \begin{align}
    \abs*{\phi^{(1)}(p)} \ge Wp^{-1} + c'_1,
  \end{align}
  then there is no consistent estimator, i.e., $R^*(n,k;\phi) \gtrsim 1$.
\end{proposition}
The proof of \cref{prop:no-estimator} is given in \cref{sec:proof-prop-no-estimator}. From \cref{lem:asm2lower}, the divergence speed of the first derivative is $p^{-1}$ if that of the fourth derivative is $p^{-4}$. Thus, if the divergence speed of $\phi$ is greater than $p^{-4}$, we cannot construct an estimator that consistently estimates $\theta$ for any probability measure $P \in \dom{M}_k$. Consequently, there is no need to derive the minimax optimal estimator in this case.

{\bfseries Minimax optimality for Class B.}
We derive the minimax optimal rate for $\phi$ in which the divergence speed of its fourth derivative is $p^{\alpha-4}$ for $\alpha \in (0,1)$. Thus, we make the following assumption.
\begin{assumption}\label{asm:bound}
  Suppose $\phi$ is four times continuously differentiable on $(0,1]$. For $\alpha \in (0,1)$, the divergence speed of the fourth derivative of $\phi$ is $p^{\alpha-4}$.
\end{assumption}
Note that a set of $\phi$ satisfying \cref{asm:bound} is Class B depicted in \cref{fig:overview}. The divergence speed increases as $\alpha$ decreases. Under \cref{asm:bound}, we derive the minimax optimal estimator of which the minimax rate is given by the following theorems.
\begin{theorem}
  \label{thm:minimax-rate}
  Under \cref{asm:bound} with $\alpha \in (0,1/2]$, if $n \gtrsim \frac{k^{1/\alpha}}{\ln k}$ and $k \gtrsim \ln^{\frac{4}{3}}n$,
  \begin{align}
    R^*(n,k;\phi) \asymp \frac{k^2}{(n\ln n)^{2\alpha}}.
  \end{align}
  Otherwise, there is no consistent estimator, i.e., $R^*(n,k;\phi) \gtrsim 1$.
\end{theorem}
\begin{theorem}
  \label{thm:minimax-rate2}
  Under \cref{asm:bound} with $\alpha \in (1/2,1)$, if $n \gtrsim \frac{k^{1/\alpha}}{\ln k}$
  \begin{align}
    R^*(n,k;\phi) \asymp \frac{k^2}{(n\ln n)^{2\alpha}} + \frac{k^{2-2\alpha}}{n}.
  \end{align}
  Otherwise, there is no consistent estimator, i.e., $R^*(n,k;\phi) \gtrsim 1$.
\end{theorem}
\cref{thm:minimax-rate,thm:minimax-rate2} are proved by combining the results in \cref{sec:anal-lower,sec:anal-upper}. The minimax optimal rate in \cref{thm:minimax-rate,thm:minimax-rate2} are characterized by the parameter for the divergence speed $\alpha$ from \cref{asm:bound}. From \cref{thm:minimax-rate,thm:minimax-rate2}, we can conclude that the minimax optimal rate decreases as the divergence speed increases.

The explicit estimator that achieves the optimal minimax rate shown in \cref{thm:minimax-rate,thm:minimax-rate2} are described in the next section.

\begin{remark}
\cref{asm:bound} covers $\phi(p)=p^\alpha$ for $\alpha \in (0,1)$, but does not for all existing works. For $\phi(p)=-p\ln(p)$ and $\phi(p)=p^\alpha$ with $\alpha \ge 1$, the divergence speed of these $\phi$ is lower than $p^{\alpha-4}$ for $\alpha \in (0,1)$. Indeed, the divergence speed of $\phi(p)=-p\ln(p)$ and $\phi(p)=p^\alpha$ for $\alpha\ge 1$ are $p^{-3}$ and $p^{\alpha-4}$, respectively. We can expect that the corresponding minimax rate is characterized by the divergence speed even if the divergence speed is lower than $p^{\alpha-4}$ for $\alpha \in (0,1)$. The analysis of the minimax rate for lower divergence speeds remains an open problem.
\end{remark}

\section{Estimator for $\theta$}

In this section, we describe our estimator for $\theta$ in detail. Our estimator is composed of the bias-corrected plugin estimator and the best polynomial estimator. We first describe the overall estimation procedure on the supposition that the bias-corrected plugin estimator and the best polynomial estimator are black boxes. Then, we describe the bias-corrected plugin estimator and the best polynomial estimator in detail.

For simplicity, we assume the samples are drawn from the Poisson sampling model, where we first draw $n' \sim \Poi(2n)$, and then draw $n'$ i.i.d. samples $S_{n'} = \cbrace{X_1,...,X_{n'}}$. Given the samples $S_{n'}$, we first partition the samples into two sets. We use one set of the samples to determine whether the bias-corrected plugin estimator or the best polynomial estimator should be employed, and the other set to estimate $\theta$. Let $\cbrace{B_i}_{i=1}^{n'}$ be i.i.d. random variables drawn from the Bernoulli distribution with parameter $1/2$, i.e., $\p\cbrace{B_i = 0} = \p\cbrace{B_i = 1} = 1/2$ for $i=1,...,n'$. We partition $(X_1,...,X_{n'})$ according to $(B_1,...,B_{n'})$, and construct the histograms $\tilde{N}$ and $\tilde{N}'$, which are defined as
\begin{align}
  \tilde{N}_i = \sum_{j=1}^{n'} \ind{X_j = i}\ind{B_j = 0}, \quad \tilde{N}'_i = \sum_{j=1}^{n'} \ind{X_j = i}\ind{B_j = 1}, \for i \in [n'].
\end{align}
Then, $\tilde{N}$ and $\tilde{N}'$ are independent histograms, and $\tilde{N}_i,\tilde{N}'_i \sim \Poi(np_i)$.

Given $\tilde{N}'$, we determine whether the bias-corrected plugin estimator or the best polynomial estimator should be employed for each alphabet. Let $\Delta_{n,k}$ be a threshold depending on $n$ and $k$ to determine which estimator is employed, which will be specified as in \cref{thm:upper-bound} on \cpageref{thm:upper-bound}. We apply the best polynomial estimator if $\tilde{N}'_i < 2\Delta_{n,k}$, and otherwise, i.e., $\tilde{N}'_i \ge 2\Delta_{n,k}$, we apply the bias-corrected plugin estimator. Let $\phi_{\rm poly}$ and $\phi_{\rm plugin}$ be the best polynomial estimator and the bias-corrected plugin estimator for $\phi$, respectively. Then, the estimator of $\theta$ is written as
\begin{align}
  \tilde\theta(\tilde{N}) = \sum_{i=1}^k \paren*{ \ind{\tilde{N}'_i \ge 2\Delta_{n,k}}\phi_{\rm plugin}(\tilde{N}_i) + \ind{\tilde{N}'_i < 2\Delta_{n,k}}\phi_{\rm poly}(\tilde{N}_i) }.
\end{align}
Finally, we truncate $\tilde\theta$ so that the final estimate is not outside of the domain of $\theta$.
\begin{align}
  \hat\theta(\tilde{N}) =& (\tilde\theta(\tilde{N}) \land \theta_{\rm sup}) \lor \theta_{\rm inf},
\end{align}
where $\theta_{\rm inf} = \inf_{P \in \dom{M}_k}\theta(P)$ and $\theta_{\rm sup} = \sup_{P \in \dom{M}_k}\theta(P)$. Next, we describe the details of the best polynomial estimator $\phi_{\rm poly}$ and the bias-corrected plugin estimator $\phi_{\rm plugin}$.

{\bfseries Best polynomial estimator.}
The best polynomial estimator is an unbiased estimator of the polynomial that provides the best approximation of $\phi$. Let $\cbrace{a_m}_{m=0}^L$ be coefficients of the polynomial that achieves the best approximation of $\phi$ by a degree-$L$ polynomial with range $I=[0,\frac{4\Delta_{n,k}}{n}]$, where $L$ is as specified in \cref{thm:upper-bound} on \cpageref{thm:upper-bound}. Then, the approximation of $\phi$ by the polynomial at point $p_i$ is written as
\begin{align}
 \phi_L(p_i) =& \sum_{m=0}^L a_m p_i^m. \label{eq:approx-poly}
\end{align}
From \cref{eq:approx-poly}, an unbiased estimator of $\phi_L$ can be derived from an unbiased estimator of $p_i^m$. For the random variable $\tilde{N}_i$ drawn from the Poisson distribution with mean parameter $np_i$, the expectation of the $m$th factorial moment $(\tilde{N}_i)_m = \tfrac{\tilde{N}_i!}{(\tilde{N}_i - m)!}$ becomes $(np_i)^m$. Thus, $\frac{(\tilde{N}_i)_m}{n^m}$ is an unbiased estimator of $p_i^m$. Substituting this into \cref{eq:approx-poly} gives the unbiased estimator of $\phi_L(p_i)$ as
\begin{align}
  \bar\phi_{\rm poly}(\tilde{N}_i) = \sum_{m=0}^L \frac{a_m}{n^j} (\tilde{N}_i)_m.
\end{align}
Next, we truncate $\bar\phi_{\rm poly}$ so that it is not outside of the domain of $\phi(p)$. Let $\phi_{{\rm inf},\frac{\Delta_{n,k}}{n}} = \inf_{p \in [0,\frac{\Delta_{n,k}}{n}]} \phi(p)$ and $\phi_{{\rm sup},\frac{\Delta_{n,k}}{n}} = \sup_{p \in [0,\frac{\Delta_{n,k}}{n}]} \phi(p)$. Then, the best polynomial estimator is defined as
\begin{align}
 \phi_{\rm poly}(\tilde{N}_i) = (\bar\phi_{\rm poly}(\tilde{N}_i) \land \phi_{{\rm sup},\frac{\Delta_{n,k}}{n}}) \lor \phi_{{\rm inf},\frac{\Delta_{n,k}}{n}}.
\end{align}

{\bfseries Bias-corrected plugin estimator.}
In the bias-corrected plugin estimator, we apply the bias correction of \citep{miller1955nbi}. Applying the second-order Taylor expansion to the bias of the plugin estimator gives
\begin{align}
 \Mean\bracket*{\phi\paren*{\frac{\tilde{N}_i}{n}} - \phi(p_i)}
  \approx& \Mean\bracket*{\phi^{(1)}(p_i)\paren*{\frac{\tilde{N}_i}{n} - p_i} + \frac{\phi^{(2)}(p_i)}{2}\paren*{\frac{\tilde{N}_i}{n} - p_i}^2} \\
  =& \frac{p_i\phi^{(2)}(p_i)}{2n}.
\end{align}
Thus, we include $-\frac{\tilde{N}_i\phi^{(2)}(\tilde{N}_i/n)}{2n^2}$ as a bias-correction term in the plugin estimator $\phi(\tilde{N}_i/n)$, which offsets the second-order approximation of the bias. However, we do not directly apply the bias-corrected plugin estimator to estimate $\phi(p_i)$ for two reasons. First, the derivative of $\phi$ is large near $0$, which results in a large bias, and second, $\phi(p)$ for $p > 1$ is undefined even though $\tilde{N}_i/n$ can exceed $1$. Thus, we apply the bias-corrected plugin estimator to the function $\bar\phi_{\frac{\Delta_{n,k}}{n}}$ defined below instead of $\phi$. Define
\begin{align}
  & H_L(p;\phi,a,b) \\
  =& \phi(a) + \sum_{m=1}^L \frac{\phi^{(m)}(a)}{m!}(p-a)^m(p-b)^{L+1}\sum_{\ell=0}^{L-m}\frac{(-1)^{\ell}(L+\ell)!}{\ell!L!}(a-b)^{-L-1-\ell}(p-a)^\ell \\
 =& \phi(a) + \sum_{m=1}^L\frac{\phi^{(m)}(a)}{m!}(p-a)^m\sum_{\ell=0}^{L-m}\frac{L+1}{L+\ell+1}\Beta_{\ell,L+\ell+1}\paren*{\frac{p-a}{b-a}},
\end{align}
where $\Beta_{\nu,n}(x)=\binom n\nu x^\nu(1-x)^{n-\nu}$ denotes the Bernstein basis polynomial. \sloppy Then, $H_L(p;\phi,a,b)$ denotes a function that interpolates between $\phi(a)$ and $\phi(b)$ using Hermite interpolation. From generalized Hermite interpolation~\citep{spitzbart1960generalization}, $H_L^{(i)}(a;\phi,a,b) = \phi^{(i)}(a)$ for $i = 0,...,L$ and $H_L^{(i)}(b;\phi,a,b) = 0$ for $i = 1,...,L$. The function $\bar\phi_{\frac{\Delta_{n,k}}{n}}$ is defined as
\begin{align}
 \bar\phi_{\frac{\Delta_{n,k}}{n}}(p) = \begin{dcases}
   H_4\paren*{\frac{\Delta_{n,k}}{2n};\phi,\frac{\Delta_{n,k}}{n},\frac{\Delta_{n,k}}{2n}} & \textif p \le \frac{\Delta_{n,k}}{2n}, \\
   H_4\paren*{p;\phi,\frac{\Delta_{n,k}}{n},\frac{\Delta_{n,k}}{2n}} & \textif \frac{\Delta_{n,k}}{2n} < p < \frac{\Delta_{n,k}}{n}, \\
   H_4\paren*{p;\phi,1,2} & \textif 1 < p < 2, \\
   H_4\paren*{2;\phi,1,2} & \textif p \ge 2, \\
   \phi(p) & \otherwise.
 \end{dcases}
\end{align}
From this definition, $\bar\phi_{\frac{\Delta_{n,k}}{n}} = \phi$ if $p \in [\frac{\Delta_{n,k}}{n}, 1]$. From Hermite interpolation, the function $\bar\phi_{\frac{\Delta_{n,k}}{n}}$ is four times differentiable on $\RealSet_+$ and $\bar\phi_{\frac{\Delta_{n,k}}{n}}^{(1)}(p) = ... = \bar\phi_{\frac{\Delta_{n,k}}{n}}^{(4)}(p) = 0$ for $p \le \frac{\Delta_{n,k}}{2n}$ and $p \ge 2$. By introducing $\bar\phi_{\frac{\Delta_{n,k}}{n}}$, we can bound the fourth derivative of $\bar\phi_{\frac{\Delta_{n,k}}{n}}$ using $\Delta_{n,k}$, and this enables us to control the bias with the threshold parameter $\Delta_{n,k}$. Using $\bar\phi_{\frac{\Delta_{n,k}}{n}}$ instead of $\phi$ yields the bias-corrected plugin estimator
\begin{align}
  \phi_{\rm plugin}(\tilde{N}_i) = \bar\phi_{\frac{\Delta_{n,k}}{n}}\paren*{\frac{\bar{N}_i}{n}} - \frac{\tilde{N}_i}{2n^2}\bar\phi_{\frac{\Delta_{n,k}}{n}}^{(2)}\paren*{\frac{\bar{N}_i}{n}}. \label{eq:plugin-estimator}
\end{align}

\section{Remark about Differentiability for Analysis}
Why is the minimax rate characterized by the divergence speed of the {\em fourth} derivative? Indeed, most of the results can be obtained on a weaker assumption compared to \cref{asm:bound} regarding differentiability, which is formally defined as follows.
\begin{assumption}\label{asm:bound2}
  Suppose $\phi$ is two times continuously differentiable on $(0,1]$. For $\alpha \in (0,1)$, the divergence speed of the second derivative of $\phi$ is $p^{\alpha-2}$.
\end{assumption}
\cref{asm:bound2} only requires two times continuous differentiability, whereas \cref{asm:bound} requires four times. Only the analysis of the bias-corrected plugin estimator requires \cref{asm:bound} to achieve the minimax rate due to the bias-correction term in \cref{eq:plugin-estimator}. The bias-correction term is formed as the plugin estimator of the second derivative of $\phi$, and its convergence rate is highly dependent on the smoothness of the second derivative. The smoothness of the second derivative of $\phi$ is characterized by the fourth derivative of $\phi$, and thus \cref{asm:bound} is required to derive the error bound of the bias-corrected plugin estimator. Another bias-correction method might weaken the assumption as in \cref{asm:bound2}.

\section{Analysis of Lower Bound}\label{sec:anal-lower}
In this section, we derive a lower bound for the minimax rate of $\theta$. Under \cref{asm:bound2}, we can derive the lower bound of the minimax risk as in the following theorem.
\begin{theorem}
  \label{thm:lower1}
  Under \cref{asm:bound2}, for $k \ge 3$, we have
 \begin{align}
  R^*(n,k;\phi) \gtrsim \frac{k^{2-2\alpha}}{n}.
 \end{align}
\end{theorem}
The lower bound is obtained by applying Le Cam's two-point method~(see \citep{DBLP:books/daglib/0035708}). The details of the proof of \cref{thm:lower1} can be found in \cref{sec:proofs-lower-bound}. Next, we derive another lower bound for the minimax rate.
\begin{theorem}
  \label{thm:lower2}
  Under \cref{asm:bound2}, if $n \gtrsim \frac{k^{1/\alpha}}{\ln k}$, we have
  \begin{align}
   R^*(n,k;\phi) \gtrsim \frac{k^2}{(n\ln n)^{2\alpha}},
  \end{align}
  where we need $k \gtrsim \ln^{\frac{4}{3}}n$ if $\alpha \in (0,1/2]$.
\end{theorem}
The proof is accomplished in the same manner as \citep[Proposition 3]{wu2016minimax}. The details of the proof of \cref{thm:lower2} are also found in \cref{sec:proofs-lower-bound}. Combining \cref{thm:lower1,thm:lower2}, we get the lower bounds in \cref{thm:minimax-rate,thm:minimax-rate2} as $R^*(n,k;\phi) \gtrsim \frac{k^2}{(n\ln n)^{2\alpha}} \lor \frac{k^{2-2\alpha}}{n} \gtrsim \frac{k^2}{(n\ln n)^{2\alpha}} + \frac{k^{2-2\alpha}}{n}$.

\section{Analysis of Upper Bound}\label{sec:anal-upper}

Here, we derive the upper bound for the worst-case risk of the estimator.
\begin{theorem}\label{thm:upper-bound}
  Suppose $\Delta_{n,k} = C_2\ln n$ and $L = \floor{C_1\ln n}$ where $C_1$ and $C_2$ are universal constants such that $6C_1\ln 2 + 4\sqrt{C_1C_2}(1+\ln2) < 1$ and $C_2 > 16$. Under \cref{asm:bound}, the worst-case risk of $\hat\theta$ is bounded above by
  \begin{align}
    \sup_{P \in \dom{M}_k} \Mean\bracket*{\paren*{\hat\theta\paren{\tilde{N}} - \theta(P)}^2} \lesssim \frac{k^2}{(n\ln n)^{2\alpha}} + \frac{k^{2-2\alpha}}{n},
  \end{align}
  where we need $k \gtrsim \ln^{\frac{4}{3}}n$ if $\alpha \in (0,1/2]$.
\end{theorem}
To prove \cref{thm:upper-bound}, we derive the bias and the variance of $\hat\theta$.
\begin{lemma}\label{lem:upper-ind-bias}
  Given $P \in \dom{M}_k$, for $1 \lesssim \Delta_{n,k} \le n$, the bias of $\hat\theta$ is bounded above by
  \begin{multline}
   \Bias\bracket*{\tilde\theta\paren{\tilde{N}} - \theta(P)} \lesssim
   \sum_{i = 1}^k \paren[\Bigg]{
     (e/4)^{\Delta_{n,k}} + \Bias\bracket*{\phi_{\rm plugin}(\tilde{N}_i) - \phi(p_i) }\ind{np_i > \Delta_{n,k}} \\
     + \Bias\bracket*{\phi_{\rm poly}(\tilde{N}_i) - \phi(p_i) }\ind{np_i \le 4\Delta_{n,k}} + e^{-\Delta_{n,k}/8}
   }.
 \end{multline}
\end{lemma}
\begin{lemma}\label{lem:upper-ind-var}
 Given $P \in \dom{M}_k$, for $1 \lesssim \Delta_{n,k} \le n$, the variance of $\hat\theta$ is bounded above by
 \begin{multline}
  \Var\bracket*{\tilde\theta\paren{\tilde{N}} - \theta(P)} \lesssim
   \sum_{i = 1}^k \paren[\Bigg]{
   (e/4)^{\Delta_{n,k}} + \Var\bracket*{\phi_{\rm plugin}(\tilde{N}_i) - \phi(p_i) }\ind{np_i > \Delta_{n,k}} \\
   + \Var\bracket*{\phi_{\rm poly}(\tilde{N}_i) - \phi(p_i)  }\ind{np_i \le 4\Delta_{n,k}} + e^{-\Delta_{n,k}/8} \\ + \paren*{\Bias\bracket*{\phi_{\rm plugin}(\tilde{N}_i) - \phi(p_i)} + \Bias\bracket*{\phi_{\rm poly}(\tilde{N}_i) - \phi(p_i)}}^2\ind{\Delta_{n,k} \le p_i \le 4\Delta_{n,k}}
  }.
 \end{multline}
\end{lemma}
The proofs of \cref{lem:upper-ind-bias,lem:upper-ind-var} are left to \cref{sec:proofs-upper-bounds}. As proved in \cref{lem:upper-ind-bias,lem:upper-ind-var}, the bounds on the bias and the variance of our estimator are obtained with the bias and the variance of the plugin and the best polynomial estimators for each individual alphabet. Thus, we next analyze the bias and the variance of the plugin and the best polynomial estimators.

\noindent{\bfseries Analysis of the best polynomial estimator.}
The following lemmas provide the upper bounds on the bias and the variance of the best polynomial estimator.
\begin{lemma}\label{lem:poly-bias}
 Let $\tilde{N} \sim \Poi(np)$. Given an integer $L$ and a positive real $\Delta$, let $\phi_L(p) = \sum_{m=0}^L a_mp^m$ be the optimal uniform approximation of $\phi$ by degree-$L$ polynomials on $[0,\Delta]$, and $g_L(\tilde{N}) = \sum_{m=0}^L a_m(\tilde{N})_m/n^m$ be an unbiased estimator of $\phi_L(p)$. Under \cref{asm:bound2}, we have
 \begin{align}
  \Bias\bracket*{(g_L(\tilde{N}) \land \phi_{{\rm sup},\Delta})\lor \phi_{{\rm inf},\Delta} - \phi(p) } \lesssim \sqrt{\Var\bracket*{g_L(\tilde{N}) - \phi_L(p)}} + \paren*{\frac{\Delta}{L^2}}^{\alpha}.
 \end{align}
\end{lemma}
\begin{lemma}\label{lem:poly-var}
 Let $\tilde{N} \sim \Poi(np)$. Given an integer $L$ and a positive real $\Delta \gtrsim \frac{1}{n}$, let $\phi_L(p) = \sum_{m=0}^L a_mp^m$ be the optimal uniform approximation of $\phi$ by degree-$L$ polynomials on $[0,\Delta]$, and $g_L(\tilde{N}) = \sum_{m=0}^L a_m(\tilde{N})_m/n^m$ be an unbiased estimator of $\phi_L(p)$. Assume \cref{asm:bound2}. If $p \le \Delta$ and $2\Delta^3L \le n$, we have
 \begin{align}
  \Var\bracket*{(g_L(\tilde{N}) \land \phi_{{\rm sup},\Delta})\lor \phi_{{\rm inf},\Delta} - \phi(p) } \lesssim \frac{\Delta^3 L 64^L (2e)^{2\sqrt{\Delta n L }}}{n}.
 \end{align}
\end{lemma}
The proofs of \cref{lem:poly-bias,lem:poly-var} can be found in \cref{sec:proofs-upper-bounds}.

\noindent{\bfseries Analysis of the plugin estimator.}
The following lemmas provide the upper bounds for the bias and the variance of the plugin estimator.
\begin{lemma}\label{lem:plugin-bias}
 Assume \cref{asm:bound} and $\frac{1}{n} \lesssim \Delta < p \le 1$. Let $\tilde{N} \sim \Poi(np)$. Then, we have
 \begin{align}
  \Bias\bracket*{\bar\phi_\Delta\paren*{\frac{\tilde{N}}{n}} - \frac{\tilde{N}}{2n^2}\bar\phi_\Delta^{(2)}\paren*{\frac{\tilde{N}}{n}} - \phi(p) } \lesssim \frac{1}{n^2\Delta^{2-\alpha}} + \frac{p}{n^2}.
 \end{align}
\end{lemma}
\begin{lemma}\label{lem:plugin-var}
 Assume \cref{asm:bound} and $\frac{1}{n} \lesssim \Delta < p \le 1$. Let $\tilde{N} \sim \Poi(np)$. Then, we have
 \begin{align}
   \Var\bracket*{\bar\phi_\Delta\paren*{\frac{\tilde{N}}{n}} - \frac{\tilde{N}}{2n^2}\bar\phi_\Delta^{(2)}\paren*{\frac{\tilde{N}}{n}} - \phi(p) + \frac{p\phi^{(2)}(p)}{2n}} \lesssim \frac{p^{2\alpha-1}}{n} + \frac{1}{n^4\Delta^{4-2\alpha}} + \frac{p}{n}.
 \end{align}
\end{lemma}
The proofs of \cref{lem:plugin-bias,lem:plugin-var} are left to \cref{sec:proofs-upper-bounds}.

\noindent{\bfseries Proof for the Upper Bound.}
Combining \cref{lem:upper-ind-bias,lem:upper-ind-var,lem:poly-bias,lem:poly-var,lem:plugin-bias,lem:plugin-var}, we prove \cref{thm:upper-bound}.
\begin{proof}[Proof of \cref{thm:upper-bound}]
 Set $L = \floor{C_1 \ln n}$ and $\Delta_{n,k} = C_2 \ln n$ where $C_1$ and $C_2$ are some positive constants. Substituting \cref{lem:poly-bias,lem:poly-var,lem:plugin-bias,lem:plugin-var} into \cref{lem:upper-ind-bias,lem:upper-ind-var} yields
 \begin{align}
   & \Bias\bracket*{\hat\theta\paren{\tilde{N}} - \theta(P)} \\
   \lesssim& \begin{multlined}[t][\textwidth-1.5em]
    \sum_{i = 1}^k \paren[\Bigg]{\frac{1}{n^{C_2(\ln 4 - 1)}} + \frac{1}{n^{\alpha}(\ln n)^{2-\alpha}} + \frac{p_i}{n^2} + \frac{(\ln n)^2 n^{3C_1\ln 2 + 2\sqrt{C_1C_2}(\ln 2 + 1)}}{n^2} \\ + \frac{1}{(n\ln n)^{\alpha}} + \frac{1}{n^{C_2/8}} }
   \end{multlined} \\
   \le& \begin{multlined}[t][\textwidth-1.5em]
    \frac{k}{n^{C_2(\ln 4 - 1)}} + \frac{k}{n^{\alpha}(\ln n)^{2-\alpha}} + \frac{1}{n^2} + \frac{k(\ln n)^2 n^{3C_1\ln 2 + 2\sqrt{C_1C_2}(\ln 2 + 1)}}{n^2} \\ + \frac{k}{(n\ln n)^{\alpha}} + \frac{k}{n^{C_2/8}},
   \end{multlined}
 \end{align}
 and
 \begin{align}
   & \Var\bracket*{\hat\theta\paren{\tilde{N}} - \theta(P)} \\
   \lesssim&\begin{multlined}[t]
    \sum_{i = 1}^k \paren[\Bigg]{
      \frac{1}{n^{C_2(\ln 4 - 1)}} + \ind{p_i \ge C_2\ln n/n}\frac{p_i^{2\alpha-1}}{n} + \frac{1}{n^{2\alpha}(\ln n)^{4-2\alpha}} + \frac{p_i}{n} + \frac{(\ln n)^4 n^{6C_1\ln 2 + 4\sqrt{C_1C_2}(\ln 2 + 1)}}{n^4} + \\ \frac{1}{n^{C_2/8}} +
      \paren*{\frac{1}{n^{\alpha}(\ln n)^{2-\alpha}} + \frac{p_i}{n^2} + \frac{(\ln n)^2 n^{3C_1\ln 2 + 2\sqrt{C_1C_2}(\ln 2 + 1)}}{n^2} + \frac{1}{(n\ln n)^{\alpha}}}^2
    }
   \end{multlined} \\
   \lesssim&\begin{multlined}[t]
     \frac{k}{n^{C_2(\ln 4 - 1)}} + \frac{k^{2-2\alpha}}{n}\lor\frac{k}{n^{2\alpha}\ln^{1-2\alpha}n} + \frac{k}{n^{2\alpha}(\ln n)^{4-2\alpha}} + \frac{1}{n} \\ + \frac{k(\ln n)^4 n^{6C_1\ln 2 + 4\sqrt{C_1C_2}(\ln 2 + 1)}}{n^4} + \frac{k}{n^{C_2/8}} + \frac{k}{n^{2\alpha}(\ln n)^{4-2\alpha}} + \frac{1}{n^4} \\ + \frac{k(\ln n)^4 n^{6C_1\ln 2 + 4\sqrt{C_1C_2}(\ln 2 + 1)}}{n^4} + \frac{k}{(n\ln n)^{2\alpha}},
   \end{multlined}
 \end{align}
 where we use \cref{lem:bound-sum-alpha,lem:bound-sum-alpha-minus}. For $\delta > 0$, as long as $C_2(\ln 4 - 1) \ge 2\alpha + \delta$, $6C_1\ln 2 + 4\sqrt{C_1C_2}(\ln 2 + 1) \le 3 - 2\alpha - \delta$, and $C_2/8 \ge 2\alpha + \delta$, we have
 \begin{align}
   \Bias\bracket*{\hat\theta\paren{\tilde{N}} - \theta(P)}^2 \lesssim& \frac{1}{n^4} + \frac{k^2}{n^{2\alpha + \delta}} + \frac{k^2}{(n\ln n)^{2\alpha}} \lesssim \frac{1}{n^4} + \frac{k^2}{(n\ln n)^{2\alpha}} \label{eq:last-bias}\\
   \Var\bracket*{\hat\theta\paren{\tilde{N}} - \theta(P)} \lesssim& \frac{k^{2-2\alpha}}{n}\lor\frac{k}{n^{2\alpha}\ln^{1-2\alpha}n} + \frac{k}{n^{2\alpha + \delta}} + \frac{k}{(n\ln n)^{2\alpha}} \\ \lesssim& \frac{k^{2-2\alpha}}{n}\lor\frac{k}{n^{2\alpha}\ln^{1-2\alpha}n} + \frac{k}{(n\ln n)^{2\alpha}} \label{eq:last-var}
 \end{align}
 There exist the constants $C_1$ and $C_2$ that satisfies these conditions, for example, $C_1 < 1/6\ln 2$ and $C_2 > 16$. Since $\hat\theta(\tilde{N}),\theta(P) \in [\theta_{\rm inf}, \theta_{\rm sup}]$, the bias-variance decomposition gives
 \begin{align}
  \sup_{P \in \dom{M}_k} \Mean\bracket*{\paren*{\hat\theta\paren{\tilde{N}} - \theta(P)}^2} \le& \sup_{P \in \dom{M}_k} \Mean\bracket*{\paren*{\tilde\theta\paren{\tilde{N}} - \theta(P)}^2}\\
  \le& \paren*{\Bias\bracket*{\tilde\theta\paren{\tilde{N}} - \theta(P)}}^2 + \Var\bracket*{\tilde\theta\paren{\tilde{N}} - \theta(P)}. \label{eq:dec-bias-var}
 \end{align}
 Substituting \cref{eq:last-bias,eq:last-var} into \cref{eq:dec-bias-var} yields
 \begin{align}
  \sup_{P \in \dom{M}_k} \Mean\bracket*{\paren*{\hat\theta\paren{\tilde{N}} - \theta(P)}^2} \lesssim \frac{k^{2-2\alpha}}{n}\lor\frac{k}{n^{2\alpha}\ln^{1-2\alpha}n} + \frac{k^2}{(n\ln n)^{2\alpha}}.
 \end{align}
 If $\alpha \in (0,1/2]$ and $k \gtrsim \ln^{\frac{4}{3}}$, the last term is dominated. If $\alpha \in (1/2,1)$, the term $\frac{k}{n^{2\alpha}\ln^{1-2\alpha}n}$ is dominated by $\frac{k^{2-2\alpha}}{n}$.
\end{proof}

\section*{Acknowledgment}
This work is supported by JST CREST and KAKENHI No. 16H02864.

\bibliographystyle{unsrtnat}
\bibliography{./references}

\begin{thebibliography}{35}
\providecommand{\natexlab}[1]{#1}
\providecommand{\url}[1]{\texttt{#1}}
\expandafter\ifx\csname urlstyle\endcsname\relax
  \providecommand{\doi}[1]{doi: #1}\else
  \providecommand{\doi}{doi: \begingroup \urlstyle{rm}\Url}\fi

\bibitem[Akaike(1998)]{akaike1998information}
Hirotogu Akaike.
\newblock Information theory and an extension of the maximum likelihood
  principle.
\newblock In \emph{Selected Papers of Hirotugu Akaike}, pages 199--213.
  Springer, 1998.

\bibitem[Lake and Moorman(2011)]{lake2011accurate}
Douglas~E Lake and J~Randall Moorman.
\newblock Accurate estimation of entropy in very short.
\newblock \emph{Am J Physiol Heart Circ Physiol}, 300:\penalty0 H319--H325,
  2011.

\bibitem[Nemenman et~al.(2004)Nemenman, Bialek, and van
  Steveninck]{nemenman2004entropy}
Ilya Nemenman, William Bialek, and Rob de~Ruyter van Steveninck.
\newblock Entropy and information in neural spike trains: Progress on the
  sampling problem.
\newblock \emph{Physical Review E}, 69\penalty0 (5):\penalty0 056111, 2004.

\bibitem[Gu et~al.(2005)Gu, McCallum, and Towsley]{gu2005detecting}
Yu~Gu, Andrew McCallum, and Don Towsley.
\newblock Detecting anomalies in network traffic using maximum entropy
  estimation.
\newblock In \emph{Proceedings of the 5th ACM SIGCOMM conference on Internet
  Measurement}, pages 32--32. USENIX Association, 2005.

\bibitem[Quinlan(1986)]{quinlan1986induction}
J.~Ross Quinlan.
\newblock Induction of decision trees.
\newblock \emph{Machine learning}, 1\penalty0 (1):\penalty0 81--106, 1986.

\bibitem[Peng et~al.(2005)Peng, Long, and Ding]{peng2005feature}
Hanchuan Peng, Fuhui Long, and Chris Ding.
\newblock Feature selection based on mutual information criteria of
  max-dependency, max-relevance, and min-redundancy.
\newblock \emph{IEEE Transactions on pattern analysis and machine
  intelligence}, 27\penalty0 (8):\penalty0 1226--1238, 2005.

\bibitem[Dhillon et~al.(2003)Dhillon, Mallela, and
  Modha]{dhillon2003information}
Inderjit~S Dhillon, Subramanyam Mallela, and Dharmendra~S Modha.
\newblock Information-theoretic co-clustering.
\newblock In \emph{Proceedings of the ninth ACM SIGKDD international conference
  on Knowledge discovery and data mining}, pages 89--98. ACM, 2003.

\bibitem[Wu and Yang(2016)]{wu2016minimax}
Yihong Wu and Pengkun Yang.
\newblock Minimax rates of entropy estimation on large alphabets via best
  polynomial approximation.
\newblock \emph{IEEE Transactions on Information Theory}, 62\penalty0
  (6):\penalty0 3702--3720, 2016.

\bibitem[Jiao et~al.(2015)Jiao, Venkat, Han, and Weissman]{jiao2015minimax}
Jiantao Jiao, Kartik Venkat, Yanjun Han, and Tsachy Weissman.
\newblock Minimax estimation of functionals of discrete distributions.
\newblock \emph{Information Theory, IEEE Transactions on}, 61\penalty0
  (5):\penalty0 2835--2885, 2015.

\bibitem[Acharya et~al.(2015)Acharya, Orlitsky, Suresh, and
  Tyagi]{DBLP:conf/soda/AcharyaOST15}
Jayadev Acharya, Alon Orlitsky, Ananda~Theertha Suresh, and Himanshu Tyagi.
\newblock The complexity of estimating r{\'{e}}nyi entropy.
\newblock In \emph{Proceedings of the Twenty-Sixth Annual {ACM-SIAM} Symposium
  on Discrete Algorithms, {SODA} 2015, San Diego, CA, USA, January 4-6, 2015},
  pages 1855--1869. {SIAM}, 2015.
\newblock \doi{10.1137/1.9781611973730.124}.

\bibitem[Antos and Kontoyiannis(2001)]{RSA:RSA10019}
András Antos and Ioannis Kontoyiannis.
\newblock Convergence properties of functional estimates for discrete
  distributions.
\newblock \emph{Random Structures \& Algorithms}, 19\penalty0 (3-4):\penalty0
  163--193, 2001.
\newblock ISSN 1098-2418.
\newblock \doi{10.1002/rsa.10019}.
\newblock URL \url{http://dx.doi.org/10.1002/rsa.10019}.

\bibitem[Miller(1955)]{miller1955nbi}
G.~A. Miller.
\newblock {Note on the bias of information estimates}, 1955.

\bibitem[Grassberger(1988)]{grassberger1988finite}
Peter Grassberger.
\newblock Finite sample corrections to entropy and dimension estimates.
\newblock \emph{Physics Letters A}, 128\penalty0 (6):\penalty0 369--373, 1988.

\bibitem[Zahl(1977)]{10.2307/1936227}
Samuel Zahl.
\newblock Jackknifing an index of diversity.
\newblock \emph{Ecology}, 58\penalty0 (4):\penalty0 907--913, 1977.
\newblock ISSN 00129658, 19399170.

\bibitem[Sch{\"u}rmann and Grassberger(1996)]{schurmann1996entropy}
Thomas Sch{\"u}rmann and Peter Grassberger.
\newblock Entropy estimation of symbol sequences.
\newblock \emph{Chaos: An Interdisciplinary Journal of Nonlinear Science},
  6\penalty0 (3):\penalty0 414--427, 1996.

\bibitem[Schober(2013)]{6620615}
S.~Schober.
\newblock Some worst-case bounds for bayesian estimators of discrete
  distributions.
\newblock In \emph{2013 IEEE International Symposium on Information Theory},
  pages 2194--2198, July 2013.
\newblock \doi{10.1109/ISIT.2013.6620615}.

\bibitem[Holste et~al.(1998)Holste, Grosse, and Herzel]{holste1998bayes}
D~Holste, I~Grosse, and H~Herzel.
\newblock Bayes' estimators of generalized entropies.
\newblock \emph{Journal of Physics A: Mathematical and General}, 31\penalty0
  (11):\penalty0 2551, 1998.

\bibitem[Han et~al.(2015)Han, Jiao, and Weissman]{DBLP:journals/corr/HanJW15a}
Yanjun Han, Jiantao Jiao, and Tsachy Weissman.
\newblock Does dirichlet prior smoothing solve the shannon entropy estimation
  problem?
\newblock \emph{CoRR}, abs/1502.00327, 2015.
\newblock URL \url{http://arxiv.org/abs/1502.00327}.

\bibitem[Paninski(2004)]{paninski2004estimating}
Liam Paninski.
\newblock Estimating entropy on m bins given fewer than m samples.
\newblock \emph{IEEE Transactions on Information Theory}, 50\penalty0
  (9):\penalty0 2200--2203, 2004.

\bibitem[Valiant and Valiant(2011{\natexlab{a}})]{DBLP:conf/stoc/ValiantV11}
Gregory Valiant and Paul Valiant.
\newblock Estimating the unseen: an n/log(n)-sample estimator for entropy and
  support size, shown optimal via new clts.
\newblock In Lance Fortnow and Salil~P. Vadhan, editors, \emph{Proceedings of
  the 43rd {ACM} Symposium on Theory of Computing, {STOC} 2011, San Jose, CA,
  USA, 6-8 June 2011}, pages 685--694. {ACM}, 2011{\natexlab{a}}.
\newblock \doi{10.1145/1993636.1993727}.
\newblock URL \url{http://doi.acm.org/10.1145/1993636.1993727}.

\bibitem[Valiant and Valiant(2011{\natexlab{b}})]{DBLP:conf/focs/ValiantV11}
Gregory Valiant and Paul Valiant.
\newblock The power of linear estimators.
\newblock In Rafail Ostrovsky, editor, \emph{{IEEE} 52nd Annual Symposium on
  Foundations of Computer Science, {FOCS} 2011, Palm Springs, CA, USA, October
  22-25, 2011}, pages 403--412. {IEEE} Computer Society, 2011{\natexlab{b}}.
\newblock \doi{10.1109/FOCS.2011.81}.
\newblock URL \url{http://dx.doi.org/10.1109/FOCS.2011.81}.

\bibitem[{Wu} and {Yang}(2015)]{2015arXiv150401227W}
Y.~{Wu} and P.~{Yang}.
\newblock {Chebyshev polynomials, moment matching, and optimal estimation of
  the unseen}.
\newblock \emph{ArXiv e-prints}, April 2015.

\bibitem[Bu et~al.(2016)Bu, Zou, Liang, and Veeravalli]{7541473}
Y.~Bu, S.~Zou, Y.~Liang, and V.~V. Veeravalli.
\newblock Estimation of kl divergence between large-alphabet distributions.
\newblock In \emph{2016 IEEE International Symposium on Information Theory
  (ISIT)}, pages 1118--1122, July 2016.
\newblock \doi{10.1109/ISIT.2016.7541473}.

\bibitem[Han et~al.(2016)Han, Jiao, and Weissman]{7840425}
Y.~Han, J.~Jiao, and T.~Weissman.
\newblock Minimax rate-optimal estimation of kl divergence between discrete
  distributions.
\newblock In \emph{2016 International Symposium on Information Theory and Its
  Applications (ISITA)}, pages 256--260, Oct 2016.

\bibitem[Jiao et~al.(2016)Jiao, Han, and Weissman]{7541399}
J.~Jiao, Y.~Han, and T.~Weissman.
\newblock Minimax estimation of the l1 distance.
\newblock In \emph{2016 IEEE International Symposium on Information Theory
  (ISIT)}, pages 750--754, July 2016.
\newblock \doi{10.1109/ISIT.2016.7541399}.

\bibitem[Timan et~al.(1965)Timan, Berry, and Cossar]{timan1965theory}
A-F Timan, J~Berry, and J~Cossar.
\newblock Theory of approximation of functions of a real variable.
\newblock 1965.

\bibitem[Petrushev and Popov(2011)]{petrushev2011rational}
Penco~Petrov Petrushev and Vasil~Atanasov Popov.
\newblock \emph{Rational approximation of real functions}, volume~28.
\newblock Cambridge University Press, 2011.

\bibitem[Ditzian and Totik(2012)]{ditzian2012moduli}
Zeev Ditzian and Vilmos Totik.
\newblock \emph{Moduli of smoothness}, volume~9.
\newblock Springer Science \& Business Media, 2012.

\bibitem[Achieser(2013)]{achieser2013theory}
Naum~I Achieser.
\newblock \emph{Theory of approximation}.
\newblock Courier Corporation, 2013.

\bibitem[Remez(1934)]{remez1934determination}
Eugene~Y Remez.
\newblock Sur la d{\'e}termination des polyn{\^o}mes d'approximation de
  degr{\'e} donn{\'e}e.
\newblock \emph{Comm. Soc. Math. Kharkov}, 10:\penalty0 41--63, 1934.

\bibitem[Spitzbart(1960)]{spitzbart1960generalization}
A~Spitzbart.
\newblock A generalization of hermite's interpolation formula.
\newblock \emph{The American Mathematical Monthly}, 67\penalty0 (1):\penalty0
  42--46, 1960.

\bibitem[Tsybakov(2009)]{DBLP:books/daglib/0035708}
Alexandre~B. Tsybakov.
\newblock \emph{Introduction to Nonparametric Estimation}.
\newblock Springer series in statistics. Springer, 2009.
\newblock ISBN 978-0-387-79051-0.
\newblock \doi{10.1007/b13794}.
\newblock URL \url{http://dx.doi.org/10.1007/b13794}.

\bibitem[Le~Cam(1986)]{LeCam:1986:AMS:20451}
Lucien~M Le~Cam.
\newblock \emph{Asymptotic Methods in Statistical Theory}.
\newblock Springer-Verlag New York, Inc., New York, NY, USA, 1986.
\newblock ISBN 0-387-96307-3.

\bibitem[Lepski et~al.(1999)Lepski, Nemirovski, and
  Spokoiny]{lepski1999estimation}
Oleg Lepski, Arkady Nemirovski, and Vladimir Spokoiny.
\newblock On estimation of the $l_r$ norm of a regression function.
\newblock \emph{Probability theory and related fields}, 113\penalty0
  (2):\penalty0 221--253, 1999.

\bibitem[Cai et~al.(2011)Cai, Low, et~al.]{cai2011testing}
T~Tony Cai, Mark~G Low, et~al.
\newblock Testing composite hypotheses, hermite polynomials and optimal
  estimation of a nonsmooth functional.
\newblock \emph{The Annals of Statistics}, 39\penalty0 (2):\penalty0
  1012--1041, 2011.

\end{thebibliography}

\appendix

\section{Error Rate of Best Polynomial Approximation}\label{sec:err-best-poly}
Here, we analyze the upper bound and the lower bound of the best polynomial approximation error $E_L\paren*{\phi, [0,\Delta]}$. The upper bound and the lower bound are derived as follows.
\begin{lemma}\label{lem:best-poly-approx}
  Under \cref{asm:bound2}, for $\Delta \in (0,1]$, we have
  \begin{align}
    E_L\paren*{\phi, [0,\Delta]} \lesssim \paren*{\frac{\Delta}{L^2}}^{\alpha}.
  \end{align}
\end{lemma}
\begin{lemma}\label{lem:lower-best-approx}
  Under \cref{asm:bound2}, for $\Delta \in (0,1]$ there is a positive constant $c$ such that
  \begin{align}
    \liminf_{L \to \infty} \paren*{\frac{L^2}{\Delta}}^{\alpha} E_L(\phi, [0,\Delta]) > c.
  \end{align}
\end{lemma}
Combining \cref{lem:best-poly-approx,lem:lower-best-approx}, we can conclude $E_(\phi, [0,\Delta]) \asymp \paren*{\frac{\Delta}{L^2}}^{\alpha}$. The proofs of these lemmas are given as follows.
\begin{proof}[Proof of \cref{lem:best-poly-approx}]
  Letting $\phi_\Delta(p) = \phi(\Delta x^2)$, we have $E_L\paren*{\phi, [0,\Delta]} = E_L\paren*{\phi_\Delta, [-1,1]}$. We utilize the Jackson's inequality to upper bound the best polynomial approximation error $E_L$ by using the modulus of continuity defined as
  \begin{align}
   \omega(f, \delta) = \sup_{x, y \in [-1,1]}\cbrace{\abs*{f(x)-f(y)} : \abs*{x-y} \le \delta}.
  \end{align}
  To derive the upper bound of $E_L$, we divide into two cases: $\alpha \in (0, 1/2]$ and $\alpha \in (1/2,1)$.

 {\bfseries Case $\alpha \in (0, 1/2]$.}
 From the Jackson's inequality~\citep{achieser2013theory}, there is a trigonometric polynomial $T_L$ with degree-$L$ such that \begin{align}
  \sup_{x \in [0,2\pi]}\abs*{f(x) - T_L(x)} \lesssim \sup_{x, y \in [0,2\pi]}\cbrace*{\abs*{f(x)-f(y)} : \abs*{x-y} \le \frac{1}{L}}.
 \end{align}
 By the definition of $E_L$, we have
 \begin{align}
   E_L(f, [-1,1]) =& \inf_{g \in \dom{P}_L}\sup_{x \in [-1,1]}\abs*{f(x) - g(x)} \\
    =& \inf_{g \in \dom{P}_L}\sup_{x \in [0,2\pi]}\abs*{f(\cos(x)) - g(\cos(x))} \\
    \lesssim& \sup_{x, y \in [0,2\pi]}\cbrace*{\abs*{f(\cos(x))-f(\cos(y))} : \abs*{x-y} \le \frac{1}{L}} \\
    =& \sup_{x, y \in [-1,1]}\cbrace*{\abs*{f(x)-f(y)} : \abs*{\cos^{-1}(x)-\cos^{-1}(y)} \le \frac{1}{L}} \\
    \le& \sup_{x, y \in [-1,1]}\cbrace*{\abs*{f(x)-f(y)} : \abs*{x-y} \le \frac{1}{L}} = \omega\paren*{f,\frac{1}{L}}, \label{eq:jackson-0}
 \end{align}
 where we use the fact that $\abs*{\cos^{-1}(x)-\cos^{-1}(y)} \ge \abs*{x-y}$ for $x,y \in [-1,1]$ to derive the last line. From \cref{lem:asm2lower} and the fact that $p^{\alpha-1} \ge 1$ for $p \in (0,1]$, we have $\abs*{\phi^{(1)}(p)} \le (W+\abs*{c_1})p^{\alpha-1}$ for $p \in (0,1]$. From the absolute continuousness of $\phi$ on $(0,1]$, for $x, y \in (-1,1]$ where $x \le y$ we have
   \begin{align}
     \abs*{\phi_\Delta(x) - \phi_\Delta(y)}
      \le& \int_x^y \abs*{2\Delta t \phi^{(1)}\paren*{\Delta t^2}}dt \\
      \le& 2\Delta^\alpha(W+\abs*{c_1}) \int_x^y t^{2\alpha-1} dt \\
      =& \frac{\Delta^\alpha(W+\abs*{c_1})}{\alpha}\paren*{y^{2\alpha} - x^{2\alpha}} \\
      \le& \frac{\Delta^\alpha(W+\abs*{c_1})}{\alpha}(y-x)^{2\alpha},
  \end{align}
  where the last line is obtained since $x^\beta$ for $\beta \in (0,1]$ is $\beta$-Holder continuous. This is valid for the case $x = 0$ since $\abs*{\phi_\Delta(0) - \phi_\Delta(y)} = \lim_{x \to 0} \abs*{\phi_\Delta(x) - \phi_\Delta(y)}$. Thus, we have
  \begin{align}
   \omega(\phi_\Delta, \delta) \le \frac{\Delta^\alpha(W+\abs*{c_1})}{\alpha}\delta^{2\alpha}.
  \end{align}
 Substituting this into \cref{eq:jackson-0}, we have
 \begin{align}
  E_L(\phi_\Delta, [-1,1]) \lesssim \frac{\Delta^\alpha(W+\abs*{c_1})}{\alpha}\frac{1}{L^{2\alpha}} \lesssim \paren*{\frac{\Delta}{L^2}}^\alpha.
 \end{align}

 {\bfseries Case $\alpha \in (1/2, 1)$.}
 From the Jackson's inequality~\citep{achieser2013theory}, there is a trigonometric polynomial $T_L$ with degree-$L$ such that \begin{align}
  \sup_{x \in [0,2\pi]}\abs*{f(x) - T_L(x)} \lesssim \frac{1}{L}\sup_{x, y \in [0,2\pi]}\cbrace*{\abs*{f^{(1)}(x)-f^{(1)}(y)} : \abs*{x-y} \le \frac{1}{L}}.
 \end{align}
 In the similar manner of the case $\alpha \in (0,1/2]$, we have
 \begin{align}
  E_L(\phi_\Delta, [-1,1]) =& \inf_{g \in \dom{P}_L}\sup_{x \in [0,2\pi]}\abs*{\phi_\Delta(\cos(x)) - g(\cos(x))}\\
   \lesssim& \frac{1}{L}\omega\paren*{\phi_\Delta^{(1)}, \frac{1}{L}}. \label{eq:jackson-1}
 \end{align}
 Since $p^{\alpha-2} \ge 1$ for $p \in (0,1]$ and \cref{asm:bound2}, we have $\abs*{\phi^{(2)}(p)} \le (\alpha_1 W+\abs*{c_2})p^{\alpha-2}$ for $p \in (0,1]$. From the absolute continuousness of $\phi^{(1)}$ on $(0,1]$, for $x, y \in (-1,1]$ where $x \le y$ we have
   \begin{align}
     \abs*{\phi_\Delta^{(1)}(x) - \phi_\Delta^{(1)}(y)}
      \le& \int_x^y \abs*{2\Delta \phi^{(1)}\paren*{\Delta t^2} + 4\Delta^2 t^2 \phi^{(2)}\paren*{\Delta t^2}}dt \\
      \le& \int_x^y \paren*{ 2\Delta^{\alpha}(W + \abs*{c_1})t^{2\alpha-2} + 4\Delta^{\alpha}(\alpha_1 W + \abs*{c_2}) t^{2\alpha-2} } dt \\
      =& \Delta^\alpha\frac{2(W+\abs*{c_1}) + 4(\alpha_1 W + \abs*{c_2})}{2\alpha-1}\paren*{y^{2\alpha-1} - x^{2\alpha-1}} \\
      \le& \Delta^\alpha\frac{2(W+\abs*{c_1}) + 4(\alpha_1 W + \abs*{c_2})}{2\alpha-1}(y-x)^{2\alpha-1}.
  \end{align}
  Also, we use the fact that $x^\beta$ for $\beta \in (0,1]$ is $\beta$-Holder continuous. Thus, we have
  \begin{align}
   \omega\paren*{\phi_\Delta^{(1)}, \delta} \le \Delta^\alpha\frac{2(W+\abs*{c_1}) + 4(\alpha_1 W + \abs*{c_2})}{2\alpha-1}\delta^{2\alpha-1}.
  \end{align}
 Substituting this into \cref{eq:jackson-1}, we have
 \begin{align}
  E_L(\phi_\Delta, [-1,1]) \lesssim \frac{1}{L}\Delta^\alpha\frac{2(W+\abs*{c_1}) + 4(\alpha_1 W + \abs*{c_2})}{2\alpha-1}\frac{1}{L^{1-2\alpha}} \lesssim \paren*{\frac{\Delta}{L^2}}^\alpha.
 \end{align}
\end{proof}
\begin{proof}[Proof of \cref{lem:lower-best-approx}]
  \sloppy Let $\phi_\Delta(x) = \phi\paren*{\Delta\frac{x+1}{2}}$. Then, we have $E_L(\phi, [0,\Delta]) = E_L(\phi_\Delta, [-1,1])$. To derive the lower bound of $E_L(\phi_\Delta, [-1,1])$, we introduce the second-order Ditzian-Totik modulus of smoothness~\citep{ditzian2012moduli} defined as
  \begin{align}
    \omega^2_\varphi(f,t) = \sup_{x,y \in [-1,1]}\cbrace*{\abs*{f(x)+f(y)-2f\paren*{\frac{x+y}{2}}} : \abs*{x - y} \le 2t\varphi\paren*{\frac{x+y}{2}}},
  \end{align}
  where $\varphi(x) = \sqrt{1-x^2}$. Fix $y = -1$, for $t > 0$ we have
  \begin{gather}
    \abs*{x - y} \le 2t\varphi\paren*{\frac{x+y}{2}} \iff \\
    x + 1 \le 2t\sqrt{1 - \frac{(x-1)^2}{4}} \iff \\
    \frac{(x+1)^2}{4t^2} \le 1 - \frac{(x-1)^2}{4} \iff \\
    t^{-2}(x+1)^2 + (x-1)^2 - 4 \le 0 \iff \\
    (t^{-2}+1)x^2 + 2(t^{-2}-1)x + (t^{-2}+1)-4 \le 0 \iff \\
    \paren*{x + \frac{t^{-2} - 1}{t^{-2} + 1}}^2 + 1 - \frac{4}{t^{-2}+1} - \frac{(t^{-2}-1)^2}{(t^{-2}+1)^2} \le 0 \iff \\
    \paren*{x + 1 - \frac{2}{t^{-2} + 1}}^2 \le \frac{4}{(t^{-2}+1)^2} \iff \\
    -1 \le x \le -1 + \frac{4}{t^{-2}+1}.
  \end{gather}
  Thus, we have
  \begin{align}
    \omega^2_\varphi(\phi_\Delta,t) \ge& \sup_{x}\cbrace*{\abs*{\phi_\Delta(x)+\phi_\Delta(-1)-2\phi_\Delta\paren*{\frac{x-1}{2}}} : -1 \le x \le -1 + \frac{4}{t^{-2}+1}} \\
    =& \sup_{x}\cbrace*{\abs*{\phi(\Delta x)+\phi(0)-2\phi\paren*{\frac{\Delta x}{2}}} : 0 \le x \le \frac{2}{t^{-2}+1}}
  \end{align}
  Application of the Taylor theorem gives
  \begin{align}
    \phi(\Delta x)+\phi(0)-2\phi\paren*{\frac{\Delta x}{2}}
     =& \begin{multlined}[t]
     \lambda\phi^{(1)}\paren*{\frac{\Delta x}{2}}\paren*{0 - \frac{x}{2}} + \lambda\phi^{(1)}\paren*{\frac{\Delta x}{2}}\paren*{x - \frac{x}{2}} \\ - \int_0^{\frac{x}{2}} \Delta^2\phi^{(2)}\paren*{\Delta t}\paren*{0 - t} dt + \int_{\frac{x}{2}}^{x} \Delta^2\phi^{(2)}(\Delta t)\paren*{x - t} dt
     \end{multlined} \\
     =& \int_0^{\frac{x}{2}} \Delta^2\phi^{(2)}(\Delta t)t dt + \int_{\frac{x}{2}}^{x} \Delta^2\phi^{(2)}(\Delta t)\paren*{x - t} dt.
  \end{align}
  Letting $p_0 = (\alpha_1 W/(\alpha_1 W \lor -c'_2))^{1/(2-\alpha)}$, $\abs*{\phi^{(2)}(p)} \ge \alpha_1 Wp^{\alpha-2} + c'_2 \ge 0$ for $(0,p_0]$. From continuousness of $\phi^{(2)}$, $\phi^{(2)}(x)$ has same sign in $x \in (0,p_0]$. Since $t \ge 0$ for $t \in [0,\frac{x}{2}]$ and $x - t \ge 0$ for $t \in [\frac{x}{2}, x]$, we have for $x \in (0,p_0]$
  \begin{align}
    & \abs*{\phi(\Delta x)+\phi(0)-2\phi\paren*{\frac{\Delta x}{2}}} \\
     \ge& \Delta^\alpha\alpha_1 W\paren*{\int_0^{\frac{x}{2}} t^{\alpha-2}t dt + \int_{\frac{x}{2}}^{x} t^{\alpha-2}\paren*{x - t} dt} + c'_2\Delta^2\paren*{\int_0^{\frac{x}{2}}t dt + \int_{\frac{x}{2}}^{x}\paren*{x - t} dt} \\
    =& \Delta^\alpha\alpha_1 W\paren*{\frac{x^\alpha}{\alpha2^\alpha} + \frac{x}{1-\alpha}\paren*{\frac{x^{\alpha-1}}{2^{\alpha-1}} - x^{\alpha-1}} + \frac{1}{\alpha}\paren*{\frac{x^{\alpha}}{2^{\alpha}} - x^{\alpha}}} + \frac{c'_2\Delta^2x^2}{4} \\
    =& \Delta^\alpha x^\alpha\paren*{W(2^{-\alpha}-1) + \frac{\alpha_1 W}{\alpha}(2^{1-\alpha}-1) + \frac{c'_2\Delta^{2-\alpha}}{4}x^{2-\alpha}} \gtrsim \Delta^{\alpha}x^{\alpha}.
  \end{align}
  Thus, we have for sufficiently small $t$
  \begin{align}
    \omega^2_\varphi(\phi_\Delta,t) \gtrsim \Delta^{\alpha}\paren*{\frac{2}{t^{-2}+1}}^{\alpha} \gtrsim \Delta^{\alpha}t^{2\alpha}. \label{eq:lower-omega-2}
  \end{align}

  With the definition of $\omega^2_\varphi(f,t)$, we have the converse result $\frac{1}{L^2}\sum_{m=1}^L(m+1)E_m(f, [-1,1]) \gtrsim \omega^2_\varphi(f,L^{-1})$~\citep{ditzian2012moduli}. Let $L'$ be an integer such that $L' = c_\ell L$ where $c_\ell  > 1$. Then, we have
 \begin{align}
   & E_L(\phi, [0,\Delta]) \\
    \ge& \frac{1}{L' - L}\sum_{m = L+1}^{L'} E_m(\phi, [0, \Delta]) \\
    \ge& \frac{1}{L'^2}\sum_{m = L+1}^{L'} (m+1)E_m(\phi, [0,\Delta]) \\
    \ge& \frac{1}{L'^2}\sum_{m = 0}^{L'} (m+1)E_m(\phi, [0,\Delta]) - \frac{1}{L'^2}E_0(\phi, [0,\Delta]) - \frac{1}{L'^2}\sum_{m=1}^L (m+1)E_m(\phi, [0,\Delta]). \label{eq:lower-best-approx}
 \end{align}
 From \cref{lem:est-bounded}, we have $\abs*{\phi(x) - \phi(y)} \le \frac{W}{\alpha}\Delta^{\alpha} + \abs*{c_1}\Delta$ for $x, y \in [0, \Delta]$. Substituting it and \cref{eq:lower-omega-2} into \cref{eq:lower-best-approx} and applying the converse result and \cref{lem:best-poly-approx} yields that there are constants $C > 0$ and $C' > 0$ such that
 \begin{align}
   E_L(\phi, [0,\Delta]) \ge& C\omega^2_\varphi(\phi_\Delta, L'^{-1}) - \frac{W}{L'^2\alpha}\Delta^{\alpha} - \frac{\abs*{c_1}}{L'^2}\Delta - \frac{C'}{L'^2}\sum_{m=1}^L(m+1) \paren*{\frac{\Delta}{m^2}}^{\alpha} \\
   \ge& C\frac{\Delta^{\alpha}}{L'^{2\alpha}} - \frac{W}{L'^2\alpha}\Delta^{\alpha} - \frac{\abs*{c_1}}{L'^2}\Delta - \frac{C'}{L'^2}\sum_{m=1}^L(m+1) \paren*{\frac{\Delta}{m^2}}^{\alpha} \\
   \ge& C\frac{\Delta^{\alpha}}{L'^{2\alpha}} - \frac{W}{\alpha c_\ell ^2L^{2\alpha}}\Delta^{\alpha} - \frac{\abs*{c_1}}{c_\ell ^2L^{2\alpha}}\Delta - \frac{2C'\Delta^{\alpha}}{L'^2}\sum_{m=1}^Lm^{1-2\alpha} \\
   \ge& C\frac{\Delta^{\alpha}}{L'^{2\alpha}} - \frac{W}{\alpha c_\ell ^2L^{2\alpha}}\Delta^{\alpha} - \frac{\abs*{c_1}}{c_\ell ^2L^{2\alpha}}\Delta - \frac{2C'\Delta^{\alpha}}{L'^2}\paren*{L^{2-2\alpha} \lor \int_0^L x^{1-2\alpha} dx} \\
   \ge& C\frac{\Delta^{\alpha}}{c_\ell^{2\alpha}L^{2\alpha}} - \frac{W}{\alpha c_\ell ^2L^{2\alpha}}\Delta^{\alpha} - \frac{\abs*{c_1}}{c_\ell^2L^{2\alpha}}\Delta - \frac{2C'\Delta^{\alpha}}{((2-2\alpha)\land 1) c_\ell ^2L^{2\alpha}} \\
   =& \frac{1}{c_\ell^{2\alpha}}\paren*{\frac{\Delta}{L^2}}^{\alpha}\paren*{C - \frac{W}{\alpha c_\ell^{2-2\alpha}} - \frac{\abs*{c_1}\Delta^{-\alpha}}{c_\ell^{2-2\alpha}} - \frac{2C'}{((2-2\alpha)\land 1) c_\ell ^{2-2\alpha}}}.
 \end{align}
 Thus, by taking sufficiently large $c_\ell$, there is $c > 0$ such that
 \begin{align}
   \limsup_{L\to\infty}\paren*{\frac{L^2}{\Delta}}^{\alpha}E_L(\phi, [0,\Delta]) > c.
 \end{align}
\end{proof}

\section{Proofs for Lower Bounds}\label{sec:proofs-lower-bound}
To prove \cref{thm:lower1}, the Le Cam's two-point method~(See, e.g., \citep{DBLP:books/daglib/0035708}). The consequent corollary of the Le Cam's two-point method is as follows.
\begin{corollary}\label{cor:le-cam}
  For any two probability measures $P, Q \in \dom{M}_k$, we have
  \begin{align}
    \tilde{R}^*(n,k;\phi) \ge \frac{1}{4}\paren*{\theta(P) - \theta(Q)}^2\exp\paren*{-n\KL(P, Q)},
  \end{align}
  where $\KL(P,Q)$ denotes the KL-divergence between $P$ and $Q$.
\end{corollary}
We provide the proof of \cref{thm:lower1}.
\begin{proof}[Proof of \cref{thm:lower1}]
  For $\epsilon \in (0,1/2)$. Define two probability measures on $[k]$ as
  \begin{align}
    P =& \paren*{\frac{1}{2}, \frac{1}{2(k-1)}, ..., \frac{1}{2(k-1)}}, \\
    Q =& \paren*{\frac{1}{2}(1+\epsilon), \frac{1}{2(k-1)}(1-\epsilon), ..., \frac{1}{2(k-1)}(1-\epsilon)}.
  \end{align}
  Then, the KL-divergence between $P$ and $Q$ is obtained as
  \begin{align}
    \KL(P,Q) =& -\frac{1}{2}\ln(1+\epsilon)-\frac{1}{2}\ln(1-\epsilon)
     = -\frac{1}{2}\ln\paren*{1 - \epsilon^2} \le \epsilon^2.
  \end{align}
  Applying the Taylor theorem gives that there exist $\xi_1 \in [1/2,(1+\epsilon)/2]$ and $\xi_2 \in [(1-\epsilon)/2(k-1),1/2(k-1)]$ such that
  \begin{align}
   &\theta(Q) - \theta(P)  \\
    =& \frac{1}{2}\phi^{(1)}\paren*{\frac{1}{2}}\epsilon - \frac{1}{2}\phi^{(1)}\paren*{\frac{1}{2(k-1)}}\epsilon + \frac{\phi^{(2)}(\xi_1)}{8}\epsilon^2 + \frac{\phi^{(2)}(\xi_2)}{8(k-1)}\epsilon^2.
  \end{align}
  From the reverse triangle inequality, we have
  \begin{align}
   & \abs*{\theta(Q) - \theta(P)} \\
    \ge& \frac{1}{2}\abs*{\phi^{(1)}\paren*{\frac{1}{2(k-1)}}}\epsilon - \abs*{\frac{1}{2}\phi^{(1)}\paren*{\frac{1}{2}}\epsilon + \frac{\phi^{(2)}(\xi_1)}{8}\epsilon^2 + \frac{\phi^{(2)}(\xi_2)}{8(k-1)}\epsilon^2} \\
    \ge& \frac{1}{2}\abs*{\phi^{(1)}\paren*{\frac{1}{2(k-1)}}}\epsilon - \abs*{\frac{1}{2}\phi^{(1)}\paren*{\frac{1}{2}}}\epsilon - \abs*{\frac{\phi^{(2)}(\xi_1)}{8}}\epsilon^2 - \abs*{\frac{\phi^{(2)}(\xi_2)}{8(k-1)}}\epsilon^2.
  \end{align}
  Combining \cref{asm:bound2}, \cref{lem:asm2lower}, and the fact that $\xi_1 \ge 1/2$ and $\xi_2 \ge 1/4(k-1)$ yields
  \begin{align}
    \abs*{\phi^{(1)}\paren*{\frac{1}{2(k-1)}}} \ge& W2^{1-\alpha}(k-1)^{1-\alpha} + c'_1, \\
    \abs*{\phi^{(1)}\paren*{\frac{1}{2}}} \le& W2^{1-\alpha} + c_1, \\
    \abs*{\phi^{(2)}(\xi_1)} \le& \alpha_1 W2^{2-\alpha} + c_2, \\
    \abs*{\phi^{(2)}(\xi_2)} \le& \alpha_1 W4^{2-\alpha}(k-1)^{2-\alpha} + c_2.
  \end{align}
  Consequently, we have
  \begin{multline}
   \abs*{\theta(Q) - \theta(P)} \ge W2^{-\alpha}\epsilon\paren*{(k-1)^{1-\alpha} - 1 - \alpha_1(2^{-1}+2^{1-\alpha}(k-1)^{1-\alpha})\epsilon} \\ - 2^{-1}(c_1 - c'_1)\epsilon - c_2(2^{-3}+2^{-3}(k-1)^{-1})\epsilon^2.
  \end{multline}
  Set $\epsilon = 1/\sqrt{n}$. Applying \cref{cor:le-cam}, we have
  \begin{align}
    & \tilde{R}^*(n,k;\phi) \\
    \ge& \begin{multlined}[t][.85\textwidth]
     \frac{W^2(k-1)^{2-2\alpha}}{2^{-2\alpha}n}\paren[\Bigg]{1 - \frac{1}{(k-1)^{1-\alpha}} - \frac{\alpha_1}{2(k-1)^\alpha\sqrt{n}} - \frac{\alpha_12^{1-\alpha}}{\sqrt{n}} \\
      - \frac{c_1-c'_1}{2^{1-\alpha}W(k-1)^{1-\alpha}} - \frac{2^{\alpha-3}c_2}{W(k-1)^{1-\alpha}\sqrt{n}} - \frac{2^{\alpha-3}c_2}{W(k-1)^{2-\alpha}\sqrt{n}} }^2
    \end{multlined} \\
    \gtrsim& \frac{k^{2-2\alpha}}{n}.
  \end{align}
  From \cref{lem:well-approx-poisson}, this lower bound is valid for $R^*(n,k;\phi)$.
\end{proof}

The proof of \cref{thm:lower2} is following the proof of \citep{wu2016minimax}. For $\epsilon \in (0,1)$, define the approximate probabilities by
\begin{align}
  \dom{M}_k(\epsilon) = \cbrace*{\cbrace{p_i}_{i=1}^k \in \RealSet^k_+ : \sum_{i=1}^kp_i \le 1 - \epsilon }.
\end{align}
With this definition, we define the minimax risk for $\dom{M}_k(\epsilon)$ as
\begin{align}
  \tilde{R}^*(n,k,\epsilon;\phi) = \inf_{\hat\theta}\sup_{P \in \dom{M}_k(\epsilon)}\Mean\paren*{\hat\theta(\tilde{N}) - \theta(P)}^2. \label{eq:risk-approx-prob}
\end{align}
The minimax risk of Poisson sampling can be bounded below by \cref{eq:risk-approx-prob} as
\begin{lemma}\label{lem:approx-lower}
  Under \cref{asm:bound2}, for any $k,n \in \NaturalSet$ and any $\epsilon < 1/3$,
  \begin{align}
    \tilde{R}^*(n/2,k;\phi) \ge \frac{1}{3}\tilde{R}^*(n,k,\epsilon;\phi) - 4\paren*{\frac{W}{\alpha}k^{1-\alpha}+\abs*{c_1}}^2e^{-n/32} - \frac{W^2}{\alpha^2}k^{2-2\alpha}\epsilon^{2\alpha} - c_1^2\epsilon^2.
  \end{align}
\end{lemma}
\begin{proof}[Proof of \cref{lem:approx-lower}]
  This proof is following the same manner of the proof of \citep[Lemma 1]{wu2016minimax}. Fix $\delta > 0$. Let $\hat\theta(\cdot,n)$ be a near-minimax optimal estimator for fixed sample size $n$, i.e.,
  \begin{align}
    \sup_{P \in \dom{M}_k}\Mean\bracket*{\paren{\hat\theta(N, n) - \theta(P)}^2} \le \delta + R^*(k,n;\phi).
  \end{align}
  For an arbitrary approximate distribution $P \in \dom{M}_k(\epsilon)$, we construct an estimator
  \begin{align}
    \tilde\theta(\tilde{N}) = \hat\theta(\tilde{N}, n'),
  \end{align}
  where $\tilde{N}_i \sim \Poi(np_i)$ and $n' = \sum_i N_i$. From the triangle inequality, \cref{lem:est-bounded} and \cref{lem:bound-sum-alpha}, we have
  \begin{align}
    & \frac{1}{3}(\tilde\theta(\tilde{N}) - \theta(P))^2 \\
     \le& \frac{1}{3}\paren*{\abs*{\tilde\theta(\tilde{N}) - \theta\paren*{\frac{P}{\sum_{i=1}^k p_i}}} + \abs*{\theta\paren*{\frac{P}{\sum_{i=1}^k p_i}} - \theta(P)}}^2 \\
     \le& \frac{1}{3}\paren*{\abs*{\tilde\theta(\tilde{N}) - \theta\paren*{\frac{P}{\sum_{i=1}^k p_i}}} + \frac{W}{\alpha}\sum_{i=1}^k\abs*{\frac{p_i}{\sum_{j=1}^k p_j} - p_i}^{\alpha} + \abs*{c_1}\sum_{i=1}^k\abs*{\frac{p_i}{\sum_{j=1}^k p_j} - p_i} }^2 \\
     \le& \begin{multlined}[t][.9\textwidth]
      \frac{1}{3}\paren[\Bigg]{\abs*{\tilde\theta(\tilde{N}) - \theta\paren*{\frac{P}{\sum_{i=1}^k p_i}}} + \frac{W}{\alpha}\sum_{i=1}^k\paren*{\frac{p_i}{\sum_{j=1}^k p_j}\abs*{\sum_{j=1}^k p_j - 1}}^{\alpha} \\ + \abs*{c_1}\sum_{i=1}^k\frac{p_i}{\sum_{j=1}^k p_j}\abs*{\sum_{j=1}^k p_j - 1} }^2
     \end{multlined} \\
     \le& \frac{1}{3}\paren*{\abs*{\tilde\theta(\tilde{N}) - \theta\paren*{\frac{P}{\sum_{i=1}^k p_i}}} + \frac{W}{\alpha}\epsilon^{\alpha}\sum_{i=1}^k\paren*{\frac{p_i}{\sum_{j=1}^k p_j}}^{\alpha} + \abs*{c_1}\epsilon\sum_{i=1}^k\frac{p_i}{\sum_{j=1}^k p_j} }^2 \\
     \le& \frac{1}{3}\paren*{\abs*{\tilde\theta(\tilde{N}) - \theta\paren*{\frac{P}{\sum_{i=1}^k p_i}}} + \frac{W}{\alpha}k^{1-\alpha}\epsilon^{\alpha} + \abs*{c_1}\epsilon }^2 \label{eq:lem-approx-diff-from}\\
     \le& \paren*{\abs*{\tilde\theta(\tilde{N}) - \theta\paren*{\frac{P}{\sum_{i=1}^k p_i}}}}^2 + \frac{W^2}{\alpha^2}k^{2-2\alpha}\epsilon^{2\alpha} + c_1^2\epsilon^2. \label{eq:lem-approx-diff}
  \end{align}
  For the first term, we observe that $\tilde{N} \sim \Mul(m, \tfrac{P}{\sum p_i})$ conditioned on $n' = m$. Therefore, we have
  \begin{align}
    \Mean\paren*{\tilde\theta(\tilde{N}) - \theta\paren*{\frac{P}{\sum_{i=1}^k p_i}}}^2
     =& \sum_{m = 0}^\infty \Mean\bracket*{\paren*{\tilde\theta(\tilde{N},m) - \theta\paren*{\frac{P}{\sum_{i=1}^k p_i}}}^2 \middle| n' = m}\p\cbrace{n' = m} \\
     \le& \sum_{m = 0}^\infty \tilde{R}^*(m,k;\phi)\p\cbrace{n' = m} + \delta.
  \end{align}
  From \cref{lem:est-bounded} and \cref{lem:bound-sum-alpha}, we have
  \begin{align}
    \tilde{R}^*(m,k;\phi) \le& \sup_{P,P' \in \dom{M}_k}\paren*{\theta(P)-\theta(P')}^2 \\
     \le& \sup_{P,P' \in \dom{M}_k}\paren*{\frac{W}{\alpha}\sum_{i=1}^k\abs*{p_i - p'_i}^\alpha + \abs*{c_1}\sum_{i=1}^k\abs*{p_i-p'_i}}^2 \\
     \le& 4\sup_{P \in \dom{M}_k}\paren*{\frac{W}{\alpha}\sum_{i=1}^kp_i^\alpha + \abs*{c_1}\sum_{i=1}^kp_i}^2 \\
     \le& 4\paren*{\frac{W}{\alpha}k^{1-\alpha}+\abs*{c_1}}^2.
  \end{align}
  Note that $\tilde{R}^*(m,k;\phi)$ is a decreasing function with respect to $m$.  Since $n' \sim \Poi(n\sum_i p_i)$ and $\abs*{\sum_i p_i - 1} \le \epsilon \le 1/3$, applying Chernoff bound yields $\p\cbrace*{n' \le n/2} \le e^{-n/32}$. Thus, we have
  \begin{align}
    & \Mean\paren*{\tilde\theta(\tilde{N}) - \theta\paren*{\frac{P}{\sum_{i=1}^k p_i}}}^2 \\
     \le& \sum_{m \ge n/K} \tilde{R}^*(m,k;\phi)\p\cbrace{n' = m} + 4\paren*{\frac{W}{\alpha}k^{1-\alpha}+\abs*{c_1}}^2\p\cbrace*{n' \le n/K} + \delta \\
     \le& \tilde{R}^*(n/K,k;\phi) + 4\paren*{\frac{W}{\alpha}k^{1-\alpha}+\abs*{c_1}}^2e^{-n/32} + \delta.
  \end{align}
  The arbitrariness of $\delta$ gives the desired result.
\end{proof}
The lower bound of $\tilde{R}^*(n,k,\epsilon;\phi)$ is given by the following lemma.
\begin{lemma}\label{lem:approx-tv-lower}
  Let $U$ and $U'$ be random variables such that $U,U' \in [0,\lambda]$ and $\Mean[U]=\Mean[U']\le 1$ and $\abs*{\Mean[\theta(U) - \theta(U')]} \ge d$, where $\lambda \le k$. Let $\epsilon = 4\lambda/\sqrt{k}$. Then
  \begin{align}
    \tilde{R}^*(n,k,\epsilon;\phi) \ge \frac{d^2}{16}\paren*{\frac{7}{8} - k\TV\paren*{\Mean[\Poi(nU/k)], \Mean[\Poi(nU'/k)]} - \frac{64W^2\lambda^{2\alpha}}{\alpha^2k^{2\alpha-1}d^2} - \frac{64c_1^2\lambda^2}{kd^2}}.
  \end{align}
\end{lemma}
\begin{proof}[Proof of \cref{lem:approx-tv-lower}]
  The proof follows the same manner of the proof of \citep[Lemma 2]{wu2016minimax} expect \cref{eq:approx-tv-expect} below. Let $\beta = \Mean[U] = \Mean[U'] \le 1$. Define two random vectors
  \begin{align}
    P = \paren*{\frac{U_1}{k},...,\frac{U_k}{k}, 1-\beta}, P' = \paren*{\frac{U'_1}{k},...,\frac{U'_k}{k}, 1-\beta},
  \end{align}
  where $U_i$a nd $U'_i$ are independent copies of $U$ and $U'$, respectively. Put $\epsilon = 4\lambda/\sqrt{k}$. Define the two events:
  \begin{align}
    \event =& \bracket*{\abs*{\sum_i\frac{U_i}{k} - \beta} \le \epsilon, \abs*{\theta(P) - \Mean[\theta(P)]} \le d/4}, \\
    \event' =& \bracket*{\abs*{\sum_i\frac{U'_i}{k} - \beta} \le \epsilon, \abs*{\theta(P') - \Mean[\theta(P')]} \le d/4}.
  \end{align}
  Applying Chebyshev's inequality, the union bound, the triangle inequality and \cref{lem:est-bounded} gives
  \begin{align}
    \p\event^c \le& \p\cbrace*{\abs*{\sum_i\frac{U_i}{k} - \beta} > \epsilon} + \p\cbrace*{\abs*{\theta(P) - \Mean[\theta(P)]} > d/4} \\
    \le& \frac{\Var[U]}{k\epsilon^2} + \frac{16\sum_{i}\Var[\phi(U_i/k)]}{d^2} \\
    \le& \frac{1}{16} + \frac{16\sum_i \Mean[(\phi(U_i/k) - \phi(\beta/k))^2]}{d^2} \\
    \le& \frac{1}{16} + \frac{32\sum_i \Mean[W^2(U_i - \beta)^{2\alpha}]}{\alpha^2k^{2\alpha}d^2} + \frac{32\sum_i \Mean[c_1^2(U_i -\beta)^2]}{k^{2}d^2} \\
    \le& \frac{1}{16} + \frac{32W^2\lambda^{2\alpha}}{\alpha^2k^{2\alpha-1}d^2} + \frac{32c_1^2\lambda^2}{kd^2} \label{eq:approx-tv-expect}
  \end{align}
  By the same manner, we have
  \begin{align}
    \p\event'^c \le \frac{1}{16} + \frac{32W^2\lambda^{2\alpha}}{\alpha^2k^{2\alpha-1}d^2} + \frac{32c_1^2\lambda^2}{kd^2}.
  \end{align}
  We define two priors on the set $\dom{M}_k(\epsilon)$, the conditional distributions $\pi = P_{U|\event}$ and $\pi' = P_{U'|\event'}$. By the definition of events $\event,\event'$ and triangle inequality, we obtain that under $\pi,\pi'$,
  \begin{align}
    \abs*{\theta(P) - \theta(P')} \ge \frac{d}{2}.
  \end{align}
  By triangle inequality, we have the total variation of observations under $\pi,\pi'$ as
  \begin{align}
    \TV(P_{\tilde{N}|\event}, P_{\tilde{N}'|\event'})
     \le& \TV(P_{\tilde{N}|\event}, P_{\tilde{N}}) + \TV(P_{\tilde{N}}, P_{\tilde{N}'}) + \TV(P_{\tilde{N}'}, P_{\tilde{N}'|\event'}) \\
     =& \p\event^c + \TV(P_{\tilde{N}}, P_{\tilde{N}'}) + \p\event'^c \\
     \le& \TV(P_{\tilde{N}}, P_{\tilde{N}'}) + \frac{1}{8} + \frac{64W^2\lambda^{2\alpha}}{\alpha^2k^{2\alpha-1}d^2} + \frac{64c_1^2\lambda^2}{kd^2}.
  \end{align}
  From the fact that total variation of product distribution can be upper bounded by the summation of individual ones, we obtain
  \begin{align}
    \TV(P_{\tilde{N}}, P_{\tilde{N}'})
     \le& \sum_{i=1}^k\TV(P_{\tilde{N}_i}, P_{\tilde{N}'_i}) + \TV(n(1-\beta), n(1-\beta)) \\
     =& k\TV(\Mean[\Poi(nU/k)],\Mean[\Poi(nU'/k)]).
  \end{align}
  Then, applying Le Cam's lemma~\citep{LeCam:1986:AMS:20451} yields that
  \begin{align}
    \tilde{R}^*(n,k,\epsilon;\phi) \ge \frac{d^2}{16}\paren*{\frac{7}{8} - k\TV\paren*{\Mean[\Poi(nU/k)], \Mean[\Poi(nU'/k)]} - \frac{64W^2\lambda^{2\alpha}}{\alpha^2k^{2\alpha-1}d^2} - \frac{64c_1^2\lambda^2}{kd^2}}.
  \end{align}
\end{proof}
To derive the upper bound of $\TV\paren*{\Mean[\Poi(nU/k)], \Mean[\Poi(nU'/k)]}$, we apply the following lemma proved by \citet{wu2016minimax}.
\begin{lemma}[{\citet[Lemma 3]{wu2016minimax}}]\label{lem:tv-poi-bound}
  Let $V$ and $V'$ be random variables on $[0,M]$. If $\Mean[V^j] = \Mean[V'^j]$, $j = 1,...,L$ and $L > 2eM$, then
  \begin{align}
    \TV(\Mean[\Poi(V)], \Mean[\Poi(V')]) \le \paren*{\frac{2eM}{L}}^L.
  \end{align}
\end{lemma}
Under the condition of \cref{lem:tv-poi-bound}, the following lemmas provides the lower bound of $d$.
\begin{lemma}\label{lem:best-approx-solution}
  For any given integer $L > 0$, there exists two probability measures $\nu_0$ and $\nu_1$ on $[0,\lambda]$ such that
  \begin{gather}
    \Mean_{X \sim \nu_0}[X^m] = \Mean_{X \sim \nu_1}[X^m], \for m=0,...,L, \\
    \Mean_{X \sim \nu_0}[\phi(X)] - \Mean_{X \sim \nu_1}[\phi(X)] = 2E_L(\phi, [0,\lambda]).
  \end{gather}
\end{lemma}
\begin{proof}[\cref{lem:best-approx-solution}]
  The proof is almost same as the proof of \citet[Lemma 10]{jiao2015minimax}. It follows directly from a standard functional analysis argument proposed by \citet{lepski1999estimation}. It suffices to replace $x^\alpha$ with $\phi(x)$ and $[0,1]$ with $[0,\lambda]$ in the proof of \citep[Lemma 1]{cai2011testing}.
\end{proof}
As proved \cref{lem:best-approx-solution}, we can choose the probability measures of $U$ and $U'$ in \cref{lem:approx-lower} so that $d$ in \cref{lem:approx-lower} becomes the uniform approximation error of the best polynomial $E_L(\phi,[0,\lambda])$. The analysis of the lower bound on $E_L(\phi,[0,\lambda])$ can be found in \cref{sec:err-best-poly}. By using the lower bound~(in \cref{lem:lower-best-approx}), we prove \cref{thm:lower2} as follows.
\begin{proof}[Proof of \cref{thm:lower2}]
  Set $L = \floor{C_1 \ln n}$ and $\lambda = C_2\frac{\ln n}{n}$ where $C_1$ and $C_2$ are universal constants such that $2eC_2 \le C_1$. Assembling \cref{lem:approx-tv-lower,lem:tv-poi-bound,lem:best-approx-solution,lem:lower-best-approx}, we have $M = C_2\frac{\ln n}{k}$, $\abs*{\Mean[\phi(U) - \phi(U')]} = d \ge ck\paren*{\tfrac{\lambda}{L^2}}^{\alpha}$ where $c > 0$ is an universal constant. Also, we have
  \begin{align}
    & \tilde{R}^*(n,k,\epsilon;\phi)\\
    \ge& \frac{d^2}{16}\paren*{\frac{7}{8} - k\paren*{\frac{2e C_2 \ln n}{k\floor{C_1 \ln n}}}^{\floor{C_1 \ln n}} - \frac{64W^2\floor{C_1 \ln n}^{4\alpha}}{c^2\alpha^2k^{2\alpha+1}} - \frac{64c_1^2\floor{C_1 \ln n}^{4\alpha}\lambda^{2-2\alpha}}{c^2k^3}}.
  \end{align}
  If $\alpha \in (1/2,1)$, it is sufficient to prove \cref{thm:lower2} when $k \gtrsim n^{1-1/2\alpha}\ln n$ because of \cref{thm:lower1}. Hence,
  \begin{align}
    \frac{64W^2\floor{C_1 \ln n}^{4\alpha}}{c^2\alpha^2k^{2\alpha+1}} =& o(1) \label{eq:o1-third}\\
    \frac{64c_1^2\floor{C_1 \ln n}^{4\alpha}\lambda^{2-2\alpha}}{c^2k^3} =& o(1). \label{eq:o1-fourth}
  \end{align}
  If $\alpha \in (0,1/2]$, we assume $k \gtrsim \ln^{\frac{4}{3}}n$. Then, we get \cref{eq:o1-third,eq:o1-fourth}. Moreover, for sufficiently large $C_1$, we get $k\paren*{\frac{2e C_2 \ln n}{k\floor{C_1 \ln n}}}^{\floor{C_1 \ln n}} = o(1)$.Thus, we have
  \begin{align}
    \tilde{R}^*(n,k,\epsilon;\phi) \gtrsim d^2 \gtrsim \frac{k^2}{(n\ln n)^{2\alpha}}. \label{eq:epsilon-order}
  \end{align}
  The second term in \cref{lem:approx-lower} is bounded above as
  \begin{align}
    4\paren*{\frac{W}{\alpha}k^{1-\alpha}+\abs*{c_1}}^2e^{-n/32} = o\paren*{\frac{k^2}{(n\ln n)^{2\alpha}}}.
  \end{align}
  For $\alpha \in (0,1)$, we get an upper bound on the fourth term in \cref{lem:approx-lower} as
  \begin{align}
    c_1^2\epsilon^2 \le& \frac{c_1^2\lambda^{2-2\alpha}L^{4\alpha}}{k^2}\cdot d^2 \\
    \le& \frac{c_1^2\lambda^{2-2\alpha}\floor{C_1 \ln n}^{4\alpha}}{k^2}\cdot d^2 = o(1)\cdot d^2.
  \end{align}
  If $\alpha \in (1/2,1)$, the third term in \cref{lem:approx-lower} is bounded above as
  \begin{align}
    \frac{W^2}{\alpha^2}k^{2-2\alpha}\epsilon^{2\alpha}
     \le& \frac{W^2L^{4\alpha}}{c^2\alpha^2k^{3\alpha}}\cdot d^2 \\
     \le& \frac{W^2\floor{C_1 \ln n}^{4\alpha}}{c^2\alpha^2k^{3\alpha}}\cdot d^2 = o(1)\cdot d^2.
  \end{align}
  Then, \cref{eq:epsilon-order,lem:approx-lower} gives
  \begin{align}
    \tilde{R}^*(n,k;\phi) \gtrsim \frac{k^2}{(n\ln n)^{2\alpha}}.
  \end{align}
  If $\alpha \in (0,1/2]$, we assume $k \ge c'\ln^{\frac{4}{3}}n$ for an arbitrary constant $c' > 0$, and we get
  \begin{align}
    \frac{W^2}{\alpha^2}k^{2-2\alpha}\epsilon^{2\alpha}
     \le& \frac{W^2C_1^{4\alpha}}{c^2\alpha^2{c'}^{3\alpha}}\cdot d^2.
  \end{align}
  Hence, for sufficiently small $c'$, \cref{eq:epsilon-order,lem:approx-lower} yields
  \begin{align}
    \tilde{R}^*(n,k;\phi) \gtrsim \frac{k^2}{(n\ln n)^{2\alpha}}.
  \end{align}
\end{proof}

\section{Proofs for Upper Bounds}\label{sec:proofs-upper-bounds}

We use the following helper lemma for proving \cref{lem:upper-ind-var}.
\begin{lemma}[\citet{cai2011testing}, Lemma 4]\label{lem:cai-lem4}
  Suppose $\ind{\event}$ is an indicator random variable independent of $X$ and $Y$, then
  \begin{align}
    \Var\bracket*{X\ind{\event} + Y\ind{\event^c}} = \Var\bracket*{X}\p\event + \Var\bracket*{Y}\p\event^c + \paren*{\Mean\bracket{X} - \Mean\bracket{Y}}^2\p\event\p\event^c.
  \end{align}
\end{lemma}
\begin{proof}[Proof of \cref{lem:upper-ind-bias}]
  From the property of the absolute value, the bias is bounded above as
  \begin{align}
    & \Bias\bracket*{ \hat\theta\paren{\tilde{N}} - \theta(P) } \\
     \le& \sum_{i=1}^k \paren*{\Bias\bracket*{\ind{\tilde{N}'_i \ge 2\Delta_{n,k}}\paren*{\phi_{\rm plugin}(\tilde{N}_i) - \phi(p_i)} } + \Bias\bracket*{\ind{\tilde{N}'_i < 2\Delta_{n,k}}\paren*{\phi_{\rm poly}(\tilde{N}_i) - \phi(p_i)}} }.
  \end{align}
  Because of the independence between $\tilde{N}$ and $\tilde{N}'$, we have
  \begin{align}
    \Bias\bracket*{\ind{\tilde{N}'_i \ge 2\Delta_{n,k}}\paren*{\phi_{\rm plugin}(\tilde{N}_i) - \phi(p_i)} }
     =& \Bias\bracket*{\phi_{\rm plugin}(\tilde{N}_i) - \phi(p_i) }\p\cbrace*{\tilde{N}'_i \ge 2\Delta_{n,k}} \\
    \Bias\bracket*{\ind{\tilde{N}'_i < 2\Delta_{n,k}}\paren*{\phi_{\rm poly}(\tilde{N}_i) - \phi(p_i)} }
     =& \Bias\bracket*{ \phi_{\rm poly}(\tilde{N}_i) - \phi(p_i) }\p\cbrace*{\tilde{N}'_i < 2\Delta_{n,k}}
  \end{align}
  For $p \in [\frac{\Delta_{n,k}}{2n}, \frac{\Delta_{n,k}}{n}]$, from \cref{lem:est-bounded,lem:asm2lower}, we have
  \begin{align}
    & \abs*{H_4\paren*{p;\phi,\frac{\Delta_{n,k}}{n},\frac{\Delta_{n,k}}{2n}} - \phi(p_i)} \\
    \le& \begin{multlined}[t][.9\textwidth]
     \abs*{\sum_{m=1}^4\frac{\phi^{(m)}\paren*{\frac{\Delta_{n,k}}{n}}}{m!}\paren*{p-\frac{\Delta_{n,k}}{n}}^m\sum_{\ell=0}^{4-m}\frac{4+1}{4+\ell+1}\Beta_{\ell,4+\ell+1}\paren*{\frac{p-\frac{\Delta_{n,k}}{n}}{\frac{\Delta_{n,k}}{2n}-\frac{\Delta_{n,k}}{n}}}} \\ + \abs*{\phi\paren*{\frac{\Delta_{n,k}}{n}} - \phi(p_i)}
    \end{multlined} \\
    \le& \begin{multlined}[t][.9\textwidth]
     \sum_{m=1}^4\frac{\abs*{\phi^{(m)}\paren*{\frac{\Delta_{n,k}}{n}}}}{m!}\paren*{\frac{\Delta_{n,k}}{2n}}^m\sum_{\ell=0}^{4-m}\binom {4+\ell}{\ell}\paren*{\frac{\ell}{4+\ell+1}}^\ell\paren*{\frac{4+1}{4+\ell+1}}^{4+1} \\ + \frac{W}{\alpha} + \abs*{c_1}
    \end{multlined} \\
    \le& \sum_{m=1}^4\frac{\abs*{\phi^{(m)}\paren*{\frac{\Delta_{n,k}}{n}}}}{m!}\paren*{\frac{\Delta_{n,k}}{2n}}^m(5-m) + \frac{W}{\alpha} + \abs*{c_1} \\
    \le& 5\sum_{m=1}^4\paren*{\frac{\alpha_{m-1}W}{m!}\paren*{\frac{\Delta_{n,k}}{n}}^{\alpha}2^{-m} + c_m\paren*{\frac{\Delta_{n,k}}{2n}}^m} + \frac{W}{\alpha} + \abs*{c_1},
  \end{align}
  where we use $0 \le \Beta_{\nu,n}(x) \le \Beta_{\nu,n}(\nu/n)$ to get the third line. From the assumption $\Delta_{n,k} \le n$, we have
  \begin{align}
   \abs*{H_4\paren*{p;\phi,\frac{\Delta_{n,k}}{n},\frac{\Delta_{n,k}}{2n}} - \phi(p)} \le 5\sum_{m=1}^4\paren*{\frac{\alpha_{m-1}W}{m!2^m} + c_m} + \frac{W}{\alpha} + \abs*{c_1}.
  \end{align}
  Also, for $p \in [1,2]$, we have
  \begin{align}
    &\abs*{H_4\paren*{p;\phi,1,2} - \phi(p_i)} \\
    \le& \abs*{\sum_{m=1}^4\frac{\phi^{(m)}(1)}{m!}(p-1)^m\sum_{\ell=0}^{4-m}\frac{4+1}{4+\ell+1}\Beta_{\ell,4+\ell+1}\paren*{\frac{p-1}{2-1}}} + \abs*{\phi(1) - \phi(p_i)} \\
    \le& 5\sum_{m=1}^4\frac{\abs*{\phi^{(m)}(1)}}{m!} + \frac{W}{\alpha} + \abs*{c_1} \\
    \le& 5\sum_{m=1}^4\paren*{\alpha_{m-1}W + c_m} + \frac{W}{\alpha} + \abs*{c_1}.
  \end{align}
  For $p \in (\frac{\Delta_{n,k}}{n},1)$,  we have by \cref{lem:est-bounded} that
  \begin{align}
   \abs*{\phi(p) - \phi(p_i)} \le \frac{W}{\alpha} + \abs*{c_1}.
  \end{align}
  Consequently, we have for $p \ge 0$
  \begin{align}
    \abs*{\bar\phi_{\frac{\Delta_{n,k}}{n}}(p) - \phi(p_i)} \le 5\sum_{m=1}^4\paren*{\alpha_{m-1}W + c_m} + \frac{W}{\alpha} + \abs*{c_1} \lesssim 1. \label{eq:plugin-uni-bias1}
  \end{align}
  For $p \in (\frac{\Delta_{n,k}}{2n}, \frac{\Delta_{n,k}}{n})$,
  \begin{align}
   & \frac{p}{2n}\abs*{H_4^{(2)}\paren*{p;\phi,\frac{\Delta_{n,k}}{n},\frac{\Delta_{n,k}}{2n}} } \\
   =& \begin{multlined}[t][\textwidth-1.5em]
    \frac{p}{2n}\abs[\Bigg]{\sum_{m=1}^4\phi^{(m)}\paren*{\frac{\Delta_{n,k}}{n}}\sum_{i=0}^2\binom{2}{i}\frac{1}{((m-i)\lor0)!}\paren*{p-\frac{\Delta_{n,k}}{n}}^{(m-i)\lor0} \\ \sum_{\ell=0}^{4-m}\frac{4+1}{4+\ell+1}\Beta^{(2-i)}_{\ell,4+\ell+1}\paren*{\frac{p-\frac{\Delta_{n,k}}{n}}{-\frac{\Delta_{n,k}}{2n}}}}
   \end{multlined} \\
   =& \begin{multlined}[t][\textwidth-1.5em]
    \frac{p}{2n}\abs[\Bigg]{\sum_{m=1}^4\phi^{(m)}\paren*{\frac{\Delta_{n,k}}{n}}\sum_{i=0}^2\binom{2}{i}\frac{1}{((m-i)\lor0)!}\paren*{p-\frac{\Delta_{n,k}}{n}}^{(m-i)\lor0} \\ \sum_{\ell=0}^{4-m}\frac{(4+1)(4+\ell+1)!}{(4+\ell+1)(4+\ell-1+i)!}\sum_{j=0}^{(2-i)\land \ell}(-1)^{j}\binom{2-i}{j}\Beta_{\ell-j,4+\ell-1+i}\paren*{\frac{p-\frac{\Delta_{n,k}}{n}}{-\frac{\Delta_{n,k}}{2n}}}},
   \end{multlined}
  \end{align}
  where the last line is obtained by using the fact $\Beta^{(1)}_{\nu,n}(x)=n(\Beta_{\nu-1,n-1}(x)-\Beta_{\nu,n-1}(x))$. Again, the fact $0 \le \Beta_{\nu,n}(x) \le \Beta_{\nu,n}(\nu/n)$ gives
  \begin{align}
   & \frac{p}{2n}\abs*{H_4^{(2)}\paren*{p;\phi,\frac{\Delta_{n,k}}{n},\frac{\Delta_{n,k}}{2n}} } \\
   \le& \begin{multlined}[t][\textwidth-1.5em]
    \frac{p}{2n}\sum_{m=1}^4\abs*{\phi^{(m)}\paren*{\frac{\Delta_{n,k}}{n}}}\sum_{i=0}^2\binom{2}{i}\frac{1}{((m-i)\lor0)!}\paren*{\frac{\Delta_{n,k}}{2n}}^{(m-i)\lor0} \\ \sum_{\ell=0}^{4-m}\sum_{j=0}^{(2-i)\land \ell}\binom{2-i}{j}\frac{(4+1)(4+\ell)!}{(\ell-j)!(4-1+i+j)!}\frac{(\ell-j)^{\ell-j}(4-1+i+j)^{4-1+i+j}}{(4+\ell-1+i)^{4+\ell-1+i}}
   \end{multlined} \\
   \le& \begin{multlined}[t][\textwidth-1.5em]
    \frac{p}{2n}\sum_{m=1}^4\abs*{\phi^{(m)}\paren*{\frac{\Delta_{n,k}}{n}}}\paren[\Bigg]{\frac{5-m}{((m-2)\lor0)!}\paren*{\frac{\Delta_{n,k}}{2n}}^{(m-2)\lor0} \\ + \frac{20(5-m)}{(m-1)!}\paren*{\frac{\Delta_{n,k}}{2n}}^{m-1} + \frac{20(4+(4-m)(5-m))}{2m!}\paren*{\frac{\Delta_{n,k}}{2n}}^m}
   \end{multlined} \\
   \le& \begin{multlined}[t][\textwidth-1.5em]
    \frac{1}{n}\sum_{m=1}^4\paren*{\alpha_{m-1}W\paren*{\frac{\Delta_{n,k}}{n}}^{\alpha-m} + c_m}\paren[\Bigg]{\frac{5-m}{((m-2)\lor0)!}\paren*{\frac{\Delta_{n,k}}{2n}}^{(m-1)\lor1} \\ + \frac{20(5-m)}{(m-1)!}\paren*{\frac{\Delta_{n,k}}{2n}}^{m} + \frac{20(4+(4-m)(5-m))}{2m!}\paren*{\frac{\Delta_{n,k}}{2n}}^{m+1}}.
   \end{multlined}
  \end{align}
  From the assumption $\Delta_{n,k} \le n$, we have
  \begin{align}
   & \frac{p}{2n}\abs*{H_4^{(2)}\paren*{p;\phi,\frac{\Delta_{n,k}}{n},\frac{\Delta_{n,k}}{2n}} } \\
   \le& \begin{multlined}[t][\textwidth-1.5em]
    \frac{1}{n}\sum_{m=1}^4\paren*{\alpha_{m-1}W\paren*{\frac{\Delta_{n,k}}{n}}^{\alpha-1} + c_m}\\\paren*{\frac{(5-m)}{2^{m-1}((m-2)\lor0)!} + \frac{20(5-m)}{2^m(m-1)!} + \frac{20(4+(4-m)(5-m))}{2^{m+2}m!}}.
   \end{multlined}
  \end{align}
  From the assumption, there is a universal constant $c > 0$ such that $\Delta_{n,k} \ge c$. Thus, we have
  \begin{align}
   & \frac{p}{2n}\abs*{H_4^{(2)}\paren*{p;\phi,\frac{\Delta_{n,k}}{n},\frac{\Delta_{n,k}}{2n}} } \\
   \le& \begin{multlined}[t][\textwidth-1.5em]
    \sum_{m=1}^4\paren*{\alpha_{m-1}W\frac{c^{\alpha-1}}{n^{\alpha}} + \frac{c_m}{n}}\\\paren*{\frac{(5-m)}{2^{m-1}((m-2)\lor0)!} + \frac{20(5-m)}{2^m(m-1)!} + \frac{20(4+(4-m)(5-m))}{2^{m+2}m!}}.
   \end{multlined}
  \end{align}
  Also, for $p \in (1,2)$, we have
  \begin{align}
   & \frac{p}{2n}\abs*{H_4^{(2)}\paren*{p;\phi,1,2} } \\
   \le& \frac{1}{n}\sum_{m=1}^4\paren*{\alpha_{m-1}W + c_m}\paren[\Bigg]{\frac{5-m}{((m-2)\lor0)!} + \frac{20(5-m)}{(m-1)!} + \frac{20(4+(4-m)(5-m))}{2m!}}.
  \end{align}
  Thus, we have for $p \ge 0$
  \begin{align}
    &\abs*{\frac{p}{2n}\bar\phi_{\frac{\Delta_{n,k}}{n}}^{(2)}(p)} \\
    \le& \begin{multlined}[t][\textwidth-1.5em]
     \sum_{m=1}^4\paren*{\alpha_{m-1}W\frac{c^{\alpha-1}}{n^{\alpha}} + \frac{\abs*{c_m}}{n}}\\\paren*{1 \lor \paren*{\frac{(5-m)}{((m-2)\lor0)!} + \frac{20(5-m)}{(m-1)!} + \frac{20(4+(4-m)(5-m))}{2m!}}}.
    \end{multlined} \\
    \lesssim& \frac{1}{n^{\alpha}}. \label{eq:plugin-uni-bias2}
  \end{align}
  Combining \cref{eq:plugin-uni-bias1,eq:plugin-uni-bias2} yields for any $p_i \in [0,1]$
  \begin{align}
   \Bias\bracket*{\phi_{\rm plugin}(\tilde{N}_i) - \phi(p_i) } \lesssim 1.
  \end{align}
  Then, we have
  \begin{align}
    & \Bias\bracket*{\phi_{\rm plugin}(\tilde{N}_i) - \phi(p_i) }\p\cbrace*{\tilde{N}'_i \ge 2\Delta_{n,k}} \\
    =& \begin{multlined}[t][.9\textwidth]
     \Bias\bracket*{\phi_{\rm plugin}(\tilde{N}_i) - \phi(p_i) }\p\cbrace*{\tilde{N}'_i \ge 2\Delta_{n,k}}\ind{np_i \le \Delta_{n,k}} \\ + \Bias\bracket*{\phi_{\rm plugin}(\tilde{N}_i) - \phi(p_i) }\p\cbrace*{\tilde{N}'_i \ge 2\Delta_{n,k}}\ind{np_i > \Delta_{n,k}}
    \end{multlined} \\
    \lesssim& \p\cbrace*{\tilde{N}'_i \ge 2\Delta_{n,k}}\ind{np_i \le \Delta_{n,k}} + \Bias\bracket*{\phi_{\rm plugin}(\tilde{N}_i) - \phi(p_i) }\ind{np_i > \Delta_{n,k}}.
  \end{align}
  \sloppy The Chernoff bound for the Poisson distribution gives $\p\cbrace*{\tilde{N}'_i \ge 2\Delta_{n,k}}\ind{np_i \le \Delta_{n,k}} \le (e/4)^{\Delta_{n,k}}$. Thus, we have
  \begin{align}
    & \Bias\bracket*{\phi_{\rm plugin}(\tilde{N}_i) - \phi(p_i) }\p\cbrace*{\tilde{N}'_i \ge 2\Delta_{n,k}} \\
    \lesssim& (e/4)^{\Delta_{n,k}} + \Bias\bracket*{\phi_{\rm plugin}(\tilde{N}_i) - \phi(p_i) }\ind{np_i > \Delta_{n,k}}. \label{eq:part-plugin-bias}
  \end{align}
  Similarly, we have by the final truncation of $\phi_{\rm poly}$ and \cref{lem:est-bounded} that
  \begin{align}
   \Bias\bracket*{ \phi_{\rm poly}(\tilde{N}_i) - \phi(p_i) } \le \sup_{p \in [0,1]}\abs*{\phi(p) - \phi(p_i)} \le \frac{W}{\alpha} + \abs*{c_1}.
  \end{align}
  The Chernoff bound yields $\p\cbrace*{\tilde{N}'_i < 2\Delta_{n,k}} \le e^{-\Delta_{n,k}/8}$ for $p_i > 4\Delta_{n,k}$. Thus, we have
  \begin{align}
    & \Bias\bracket*{ \phi_{\rm poly}(\tilde{N}_i) - \phi(p_i) }\p\cbrace*{\tilde{N}'_i < 2\Delta_{n,k}} \\
    \le& \begin{multlined}[t]
     \Bias\bracket*{\phi_{\rm poly}(\tilde{N}_i) - \phi(p_i) }\p\cbrace*{\tilde{N}'_i < 2\Delta_{n,k}}\ind{np_i \le 4\Delta_{n,k}} \\ + \Bias\bracket*{\phi_{\rm poly}(\tilde{N}_i) - \phi(p_i) }\p\cbrace*{\tilde{N}'_i < 2\Delta_{n,k}}\ind{np_i > 4\Delta_{n,k}}
    \end{multlined} \\
    \le& \Bias\bracket*{\phi_{\rm poly}(\tilde{N}_i) - \phi(p_i) }\ind{np_i \le 4\Delta_{n,k}} + \paren*{\frac{W}{\alpha} + \abs*{c_1}}e^{-\Delta_{n,k}/8}. \label{eq:part-poly-bias}
  \end{align}
  Combining \cref{eq:part-plugin-bias,eq:part-poly-bias} gives the desired result.
\end{proof}
\begin{proof}[Proof of \cref{lem:upper-ind-var}]
 \sloppy Because of the independence of $\tilde{N}_1,..,\tilde{N}_k,\tilde{N}'_1,...,\tilde{N}'_k$, applying \cref{lem:cai-lem4} gives
  \begin{align}
    & \Var\bracket*{\hat\theta\paren{\tilde{N}} - \theta(P)} \\
    \le& \Var\bracket*{ \sum_{i=1}^k \ind{\tilde{N}'_i \ge 2\Delta_{n,k}}\paren*{\phi_{\rm plugin}(\tilde{N}_i) - \phi(p_i)} + \ind{\tilde{N}'_i < 2\Delta_{n,k}}\paren*{\phi_{\rm poly}(\tilde{N}_i) - \phi(p_i)} } \\
    \le& \sum_{i=1}^k \Var\bracket*{ \ind{\tilde{N}'_i \ge 2\Delta_{n,k}}\paren*{\phi_{\rm plugin}(\tilde{N}_i) - \phi(p_i)} + \ind{\tilde{N}'_i < 2\Delta_{n,k}}\paren*{\phi_{\rm poly}(\tilde{N}_i) - \phi(p_i)} } \\
   \le& \begin{multlined}[t][\textwidth-4.5em]
    \sum_{i=1}^k \paren[\Bigg]{ \Var\bracket*{\phi_{\rm plugin}(\tilde{N}_i) - \phi(p_i)}\p\cbrace*{\tilde{N}'_i \ge 2\Delta_{n,k}} \\
    + \Var\bracket*{\phi_{\rm poly}(\tilde{N}_i) - \phi(p_i)}\p\cbrace*{\tilde{N}'_i < 2\Delta_{n,k}} \\
    + \paren*{\Mean\bracket*{\phi_{\rm plugin}(\tilde{N}_i) - \phi(p_i)} - \Mean\bracket*{\phi_{\rm poly}(\tilde{N}_i) - \phi(p_i)}}^2\\\p\cbrace*{\tilde{N}'_i \ge 2\Delta_{n,k}}\p\cbrace*{\tilde{N}'_i < 2\Delta_{n,k}} }.
   \end{multlined} \label{eq:divided-var}
  \end{align}
  We can derive upper bounds on the first two terms of \cref{eq:divided-var} in the same manner of \cref{eq:part-plugin-bias,eq:part-poly-bias} as
  \begin{multline}
    \Var\bracket*{\phi_{\rm plugin}(\tilde{N}_i) - \phi(p_i)}\p\cbrace*{\tilde{N}'_i \ge 2\Delta_{n,k}} \\
     \lesssim (e/4)^{\Delta_{n,k}} + \Var\bracket*{\phi_{\rm plugin}(\tilde{N}_i) - \phi(p_i) }\ind{np_i > \Delta_{n,k}},
  \end{multline}
  and
  \begin{multline}
    \Var\bracket*{\phi_{\rm poly}(\tilde{N}_i) - \phi(p_i)}\p\cbrace*{\tilde{N}'_i < 2\Delta_{n,k}} \\
     \lesssim \Var\bracket*{\phi_{\rm poly}(\tilde{N}_i) - \phi(p_i)  }\ind{np_i \le 4\Delta_{n,k}} + e^{-\Delta_{n,k}/8}.
  \end{multline}
  By the Chernoff bound, we have
  \begin{align}
    & \p\cbrace*{\tilde{N}'_i \ge 2\Delta_{n,k}}\p\cbrace*{\tilde{N}'_i < 2\Delta_{n,k}} \\
    =& (\ind{p_i < \Delta_{n,k}} + \ind{p_i > 4\Delta_{n,k}} + \ind{\Delta_{n,k} \le p_i \le 4\Delta_{n,k}}) \p\cbrace*{\tilde{N}'_i \ge 2\Delta_{n,k}}\p\cbrace*{\tilde{N}'_i < 2\Delta_{n,k}} \\
    \le& (e/4)^{\Delta_{n,k}} + e^{-\Delta_{n,k}/8} + \ind{\Delta_{n,k} \le p_i \le 4\Delta_{n,k}}.
  \end{align}
  Thus, we have the upper bound of the last term of \cref{eq:divided-var} as
  \begin{align}
    & \paren*{\Mean\bracket*{\phi_{\rm plugin}(\tilde{N}_i) - \phi(p_i)} - \Mean\bracket*{\phi_{\rm poly}(\tilde{N}_i) - \phi(p_i)}}^2\p\cbrace*{\tilde{N}'_i \ge 2\Delta_{n,k}}\p\cbrace*{\tilde{N}'_i < 2\Delta_{n,k}} \\
    \le& \begin{multlined}[t][\textwidth-1.5em]
     \paren*{\Bias\bracket*{\phi_{\rm plugin}(\tilde{N}_i) - \phi(p_i)} + \Bias\bracket*{\phi_{\rm poly}(\tilde{N}_i) - \phi(p_i)}}^2 \\ \paren*{(e/4)^{\Delta_{n,k}} + e^{-\Delta_{n,k}/8} + \ind{\Delta_{n,k} \le p_i \le 4\Delta_{n,k}}}
    \end{multlined} \\
    \lesssim& \begin{multlined}[t][\textwidth-1.5em]
     (e/4)^{\Delta_{n,k}} + e^{-\Delta_{n,k}/8} \\ + \paren*{\Bias\bracket*{\phi_{\rm plugin}(\tilde{N}_i) - \phi(p_i)} + \Bias\bracket*{\phi_{\rm poly}(\tilde{N}_i) - \phi(p_i)}}^2\ind{\Delta_{n,k} \le p_i \le 4\Delta_{n,k}}.
    \end{multlined}
  \end{align}
\end{proof}
Next, we prove the upper bounds on the bias and the variance of the best polynomial estimator as follows:
\begin{proof}[Proof of \cref{lem:poly-bias}]
  Let $\phi'_{{\rm sup},\Delta} = \phi_{{\rm sup},\Delta} \lor \sup_{p \in [0,\Delta]}\phi_L(p)$ and $\phi'_{{\rm inf},\Delta} = \phi_{{\rm inf},\Delta} \land \inf_{p \in [0,\Delta]}\phi_L(p)$. By the triangle inequality and the fact that $g_L$ is an unbiased estimator of $\phi_L$, we have
 \begin{multline}
 \Bias\bracket*{(g_L(\tilde{N}) \land \phi_{{\rm sup},\Delta})\lor \phi_{{\rm inf},\Delta} - \phi(p) } \\
  \le \Bias\bracket*{(g_L(\tilde{N}) \land \phi_{{\rm sup},\Delta})\lor \phi_{{\rm inf},\Delta} - (g_L(\tilde{N}) \land \phi'_{{\rm sup},\Delta})\lor \phi'_{{\rm inf},\Delta} } \\ + \Bias\bracket*{(g_L(\tilde{N}) \land \phi'_{{\rm sup},\Delta})\lor \phi'_{{\rm inf},\Delta} - \phi_L(p) }  + \Bias\bracket*{g_L(\tilde{N}) - \phi(p) }.
 \end{multline}
 By Chebyshev alternating theorem~\citep{petrushev2011rational}, the first term is bounded above as
 \begin{align}
   & \Bias\bracket*{(g_L(\tilde{N}) \land \phi_{{\rm sup},\Delta})\lor \phi_{{\rm inf},\Delta} - (g_L(\tilde{N}) \land \phi'_{{\rm sup},\Delta})\lor \phi'_{{\rm inf},\Delta} } \\
   \le& (\phi'_{{\rm sup},\Delta} - \phi_{{\rm sup},\Delta})\lor(\phi_{{\rm inf},\Delta} - \phi'_{{\rm inf},\Delta}) \le E_L(\phi, [0,\Delta]).
 \end{align}
 Also, the third term is bounded above as
 \begin{align}
  \Bias\bracket*{g_L(\tilde{N}) - \phi(p) } =& \abs*{\phi_L(p) - \phi(p)} \le E_L\paren*{\phi, [0,\Delta]}.
 \end{align}
 The error bound of $E_L\paren*{\phi, [0,\Delta]}$ is derived in \cref{sec:err-best-poly}.  From \cref{lem:best-poly-approx}, we have $E_L\paren*{\phi, [0,\Delta]} \lesssim \paren*{\frac{\Delta}{L^2}}^{\alpha}$. The second term has upper bound as
 \begin{align}
 \Bias\bracket*{(g_L(\tilde{N}) \land \phi'_{{\rm sup},\Delta})\lor \phi'_{{\rm inf},\Delta} - \phi_L(p) }
  =& \sqrt{\paren*{\Mean\bracket*{(g_L(\tilde{N}) \land \phi'_{{\rm sup},\Delta})\lor \phi'_{{\rm inf},\Delta} - \phi_L(p)}}^2} \\
  \le& \sqrt{\Mean\bracket*{\paren*{(g_L(\tilde{N}) \land \phi'_{{\rm sup},\Delta})\lor \phi'_{{\rm inf},\Delta} - \phi_L(p) }^2}}.
 \end{align}
 Since $\phi_L(p) \in [\phi'_{{\rm inf},\Delta},\phi'_{{\rm sup},\Delta}]$ for $p \in [0,\Delta]$, we have $\paren*{(g_L(\tilde{N}) \land \phi'_{{\rm sup},\Delta})\lor \phi'_{{\rm inf},\Delta} - \phi_L(p) }^2 \le \paren*{g_L(\tilde{N}) - \phi_L(p) }^2$. Thus, we have
 \begin{align}
   \Bias\bracket*{(g_L(\tilde{N}) \land \phi'_{{\rm sup},\Delta})\lor \phi'_{{\rm inf},\Delta} - \phi_L(p) } \le \sqrt{\Var\bracket*{g_L(\tilde{N}) - \phi_L(p)}}.
 \end{align}
\end{proof}
\begin{proof}[Proof of \cref{lem:poly-var}]
 It is obviously that truncation does not increase the variance, i.e.,
 \begin{align}
   \Var\bracket*{(g_L(\tilde{N}) \land \phi_{{\rm sup},\Delta})\lor \phi_{{\rm inf},\Delta} - \phi(p) } \le \Var\bracket*{g_L(\tilde{N}) - \phi(p) }.
 \end{align}
 Letting $\phi_\Delta(p) = \phi(\Delta x)$ and $a_0,...,a_L$ be coefficients of the optimal uniform approximation of $\phi_\Delta$ by degree-$L$ polynomials on $[0,1]$, we have $\sum_{m=0}^L \frac{\Delta^m a_m}{n^m}(\tilde{N})_m = g_L(\tilde{N})$. Then, since the standard deviation of sum of random variables is at most the sum of individual standard deviation, we have
 \begin{align}
   \Var\bracket*{g_L(\tilde{N}) - \phi(p) }
    \le \paren*{\sum_{m=1}^L \frac{\Delta^m \abs*{a_m}}{n^m}\sqrt{\Var \paren{\tilde{N}}_m } }^2.
 \end{align}
 From \citep{petrushev2011rational} and the fact from \cref{lem:est-bounded} that $\phi$ is bounded, there is a positive constant $C$ such that $\abs*{a_m} \le C2^{3L}$. From \citep{wu2016minimax}, $\Var \paren{\tilde{N}}_m$ is decreasing monotonously as $m$ increases, and for $X \sim \Poi(\lambda)$
 \begin{align}
   \Var \paren*{X}_m \le (\lambda m)^m\paren*{\frac{(2e)^{2\sqrt{\lambda m}}}{\pi\sqrt{\lambda m}} \lor 1}.
 \end{align}
 By the assumption of $p \le \Delta$ and monotonous, $\Var\paren{\tilde{N}}_m \le \Var\paren*{X}_m$ where $X \sim \Poi(\Delta n)$. Thus, we have
 \begin{align}
   \Var\bracket*{g_L(\tilde{N})}
    \lesssim& \paren*{\sum_{m=1}^L \frac{\Delta^m 2^{3L}}{n^{m}}\sqrt{ (\Delta n L )^m(2e)^{2\sqrt{\Delta n L}} } }^2 \\
    \le& \paren*{\sum_{m=1}^L \sqrt{\frac{\Delta^{3m} L^{m}}{n^{m}}} 2^{3L}(2e)^{\sqrt{\Delta n L}} }^2.
 \end{align}
 From the assumption $\frac{\Delta^3 L}{n} \le \frac{1}{2}$, we have
 \begin{align}
   & \paren*{\sum_{m=1}^L c^m\sqrt{\frac{\Delta^{3m} L^{m}}{n^{m}}} 2^{3L}(2e)^{\sqrt{\Delta n L}} }^2 \\
   \le& \paren*{2^{3L}(2e)^{\sqrt{\Delta n L}}\sum_{m=1}^L \paren*{\sqrt{\frac{\Delta^{3} L}{n}}}^m }^2 \\
   \le& \paren*{2^{3L}(2e)^{\sqrt{\Delta n L}}\paren*{\sqrt{\frac{\Delta^{3} L}{n}} + \int_1^L \paren*{\sqrt{\frac{\Delta^{3} L}{n}}}^x dx} }^2 \\
   \le& \paren*{2^{3L}(2e)^{\sqrt{\Delta n L}}\paren*{\sqrt{\frac{\Delta^{3} L}{n}} + \frac{2}{\ln\paren*{\frac{\Delta^{3} L}{n}}}\paren*{\paren*{\sqrt{\frac{\Delta^{3} L}{n}}}^L - \sqrt{\frac{\Delta^{3} L}{n}}}} }^2 \\
   =& \paren*{\sqrt{\frac{\Delta^{3} L}{n}}2^{3L}(2e)^{\sqrt{\Delta n L}}\paren*{1 + \frac{2}{\ln 2}\paren*{1 - \paren*{\sqrt{\frac{\Delta^{3} L}{n}}}^{L-1}}} }^2 \\
   \le& \frac{16\Delta^3L64^L(2e)^{2\sqrt{\Delta n L}}}{n} \\
   \lesssim& \frac{\Delta^3 L 64^L (2e)^{2\sqrt{\Delta n L }}}{n}.
 \end{align}
\end{proof}

The proofs of the upper bounds on the bias and the variance of the bias-corrected plugin estimator are obtained as follows.
\begin{proof}[Proof of \cref{lem:plugin-bias}]
Applying Taylor theorem yields
 \begin{align}
   &\Bias\bracket*{\bar\phi_{\Delta}\paren*{\frac{\tilde{N}}{n}} - \phi(p) } \\
    =& \begin{multlined}[t][\textwidth-1.5em]
     \abs[\Bigg]{\Mean\bracket[\Bigg]{\phi^{(1)}(p)\frac{\tilde{N} - np}{n} + \frac{\phi^{(2)}(p)}{2}\paren*{\frac{\tilde{N}}{n} - p}^2 - \frac{\tilde{N}}{2n}\bar\phi_{\Delta}^{(2)}\paren*{\frac{\tilde{N}}{n}} \\
     +  \frac{\phi^{(3)}(p)}{6}\paren*{\frac{\tilde{N}}{n} - p}^3 + R_3\paren*{\frac{\tilde{N}}{n};\bar\phi_\Delta,p} }}
    \end{multlined} \\
    \le& \frac{1}{2n}\abs*{\Mean\bracket*{p\phi^{(2)}(p) - \frac{\tilde{N}}{n}\bar\phi_{\Delta}^{(2)}\paren*{\frac{\tilde{N}}{n}}}} + \frac{p\abs*{\phi^{(3)}(p)}}{6n^2} + \abs*{\Mean\bracket*{R_3\paren*{\frac{\tilde{N}}{n};\bar\phi_\Delta,p}} }, \label{eq:plugin-bias-1}
\end{align}
where we use the fact that for $X \sim \Poi(\lambda)$, $\Mean[(X-\lambda)^2] = \lambda$, $\Mean[(X-\lambda)^3] = \lambda$, and $R_3(x;\bar\phi_\Delta,p)$ denotes the reminder term of the Taylor theorem. The first term of \cref{eq:plugin-bias-1} is bounded above as
\begin{align}
 &\frac{1}{2n}\abs*{\Mean\bracket*{p\phi^{(2)}(p) - \frac{\tilde{N}}{n}\bar\phi_{\Delta}^{(2)}\paren*{\frac{\tilde{N}}{n}}}} \\
  =& \frac{1}{2n}\abs*{\Mean\bracket*{ \phi^{(2)}(p)\paren*{p - \frac{\tilde{N}}{n}} +  \frac{\tilde{N}}{n}\paren*{\phi^{(2)}(p) - \bar\phi_{\Delta}^{(2)}\paren*{\frac{\tilde{N}}{n}}} }} \\
  =& \frac{1}{2n}\abs*{\Mean\bracket*{ \frac{\tilde{N}\phi^{(3)}(p)}{n}\paren*{ \frac{\tilde{N}}{n} - p} + \frac{\tilde{N}}{n}R_1\paren*{\frac{\tilde{N}}{n};\bar\phi_\Delta^{(2)},p}}} \\
  \le& \frac{p\abs*{\phi^{(3)}(p)}}{2n^2} + \abs*{\Mean\bracket*{\frac{\tilde{N}}{2n^2}R_1\paren*{\frac{\tilde{N}}{n};\bar\phi_\Delta^{(2)},p}}}, \label{eq:plugin-bias-2}
\end{align}
where the last line is obtained by using the fact that for $X \sim \Poi(\lambda)$, $\Mean[X(X-\lambda)] = \lambda$, and $R_1(x;\bar\phi_\Delta^{(2)},p)$ denotes the reminder term of the Taylor theorem. From \cref{lem:asm2lower}, the second term of \cref{eq:plugin-bias-1} and the first term of \cref{eq:plugin-bias-2} are bounded above as
\begin{align}
  \frac{p\abs*{\phi^{(3)}(p)}}{6n^2} \le& \frac{\alpha_2Wp^{\alpha-2}+c_3p}{6n^2} \lesssim \frac{1}{n^2\Delta^{2-\alpha}} + \frac{p}{n^2} \label{eq:plugin-bias-r1} \\
  \frac{p\abs*{\phi^{(3)}(p)}}{2n^2} \le& \frac{\alpha_2Wp^{\alpha-2}+c_3p}{2n^2} \lesssim \frac{1}{n^2\Delta^{2-\alpha}} + \frac{p}{n^2}. \label{eq:plugin-bias-r2}
\end{align}
The rest is to derive the upper bound on $\abs*{\Mean\bracket*{R_3\paren*{\frac{\tilde{N}}{n};\bar\phi_\Delta,p}}}$ and $\abs*{\Mean\bracket*{\frac{\tilde{N}}{2n^2}R_1\paren*{\frac{\tilde{N}}{n};\bar\phi_\Delta^{(2)},p}}}$. Let $\hat{p} = \frac{\tilde{N}}{n}$. From the mean value theorem, letting a function $G(x)$ be continuous on the closed interval and differentiable with non-vanishing derivative on the open interval between $p$ and $\hat{p}$, there exists $\xi$ between $p$ and $\hat{p}$ such that
\begin{align}
  R_3\paren*{\hat{p};\bar\phi_\Delta,p} =& \frac{\bar\phi_\Delta^{(4)}(\xi)}{6}(\hat{p}-\xi)^3\frac{G(\hat{p}) - G(p)}{G^{(1)}(\xi)}.
\end{align}
Define $G(x) = \frac{1}{x^2}(\hat{p}-x)^4$. Then, there exists $\xi$ such that
\begin{align}
  R_3\paren*{\hat{p};\bar\phi_\Delta,p} =& -\frac{\bar\phi_\Delta^{(4)}(\xi)}{12}(\hat{p}-\xi)^3\frac{\xi^3(\hat{p}-p)^4}{p^2(\xi+\hat{p})(\hat{p}-\xi)^3} \\
  =& -\frac{\xi^3\bar\phi_\Delta^{(4)}(\xi)}{12p^2(\xi+\hat{p})}(\hat{p}-p)^4 \label{eq:r3-phi-reminder}
\end{align}
Thus, we have
\begin{align}
  \abs*{\Mean\bracket*{R_3\paren*{\hat{p};\bar\phi_\Delta,p}}}
   \le& \Mean\bracket*{\frac{\xi^3\abs*{\bar\phi_\Delta^{(4)}(\xi)}}{12p^2(\xi+\hat{p})}(\hat{p}-p)^4} \\
   \le& \frac{1}{12p^2}\Mean\bracket*{\xi^2\abs*{\bar\phi_\Delta^{(4)}(\xi)}(\hat{p}-p)^4} \\
   \le& \frac{\sup_{\xi \in \RealSet_+}\xi^2\abs*{\bar\phi_\Delta^{(4)}(\xi)}}{12p^2}\Mean\bracket*{(\hat{p}-p)^4} \\
   \le& \paren*{\frac{1}{4n^2} + \frac{1}{12pn^3}}\sup_{\xi \in \RealSet_+}\xi^2\abs*{\bar\phi_\Delta^{(4)}(\xi)},
\end{align}
where we use the fact that for $X \sim \Poi(\lambda)$, $\Mean[X^4] = 3\lambda^2+\lambda$. For $\xi \in (\frac{\Delta}{2},\Delta)$, we have
\begin{align}
  & \abs*{H^{(4)}\paren*{\xi;\phi,\frac{\Delta}{2},\Delta}} \\
  =& \abs*{\sum_{m=1}^4\phi^{(m)}(\Delta)\sum_{i=0}^4\binom{4}{i}\frac{1}{((m-i)\lor 0)!}(\xi-\Delta)^{(m-i)\lor 0} \sum_{\ell=0}^{4-m}\frac{5}{5+\ell}\Beta^{(4-i)}_{\ell,5+\ell}\paren*{\frac{\xi-\Delta}{\Delta/2}} } \\
  =& \begin{multlined}[t][\textwidth-1.5em]
   \abs[\Bigg]{\sum_{m=1}^4\phi^{(m)}(\Delta)\sum_{i=0}^4\binom{4}{i}\frac{1}{((m-i)\lor 0)!}(\xi-\Delta)^{(m-i)\lor 0}\\ \sum_{\ell=0}^{4-m}\frac{5(5+\ell)!}{(5+\ell)(1+\ell+i)}\sum_{j=0}^{(4-i)\land\ell}(-1)^j\binom{4-i}{j}\Beta_{\ell-j,1+\ell+i}\paren*{\frac{\xi-\Delta}{\Delta/2}} },
  \end{multlined}
\end{align}
where we use $\Beta^{(1)}_{\nu,n}(x)=n(\Beta_{\nu-1,n-1}(x)-\Beta_{\nu,n-1}(x))$. Since $0 \le \Beta_{\nu,n}(x) \le \Beta_{\nu,n}(\nu/n) \le 1$, there is a universal constant $c > 0$ such that for any $i=0,...,4$
\begin{align}
  \abs*{\sum_{\ell=0}^{4-m}\frac{5(5+\ell)!}{(5+\ell)(1+\ell+i)}\sum_{j=0}^{(4-i)\land\ell}(-1)^j\binom{4-i}{j}\Beta_{\ell-j,1+\ell+i}\paren*{\frac{\xi-\Delta}{\Delta/2}}} \le c.
\end{align}
Thus, we have from \cref{lem:asm2lower} that
\begin{align}
  & \xi^2\abs*{H^{(4)}\paren*{\xi;\phi,\frac{\Delta}{2},\Delta}} \\
  \le& \sum_{m=1}^4\abs*{\phi^{(m)}(\Delta)}\sum_{i=0}^4\binom{4}{i}\frac{c}{((m-i)\lor 0)!}\abs*{\xi^2\paren*{\xi-\Delta}^{(m-i)\lor 0}} \\
  \le& \sum_{m=1}^4\paren*{\alpha_{m-1}W\Delta^{\alpha-m}+c_m}\sum_{i=0}^4\binom{4}{i}\frac{c}{((m-i)\lor 0)!}\Delta^{(2+m-i)\lor 2} \\
  =&  \sum_{m=1}^4\sum_{i=0}^4\binom{4}{i}\frac{c}{((m-i)\lor 0)!}\paren*{\alpha_{m-2}W\Delta^{(2+\alpha-i)\lor(2+\alpha-m)} + c_m\Delta^{(2+m-i)\lor 2}} \\
  \lesssim& \Delta^{\alpha-2}.
\end{align}
Similarly, for $\xi \in (1,2)$
\begin{align}
  & \xi^2\abs*{H^{(4)}\paren*{\xi;\phi,2,1}} \\
  \le& \sum_{m=1}^4\abs*{\phi^{(m)}(1)}\sum_{i=0}^4\binom{4}{i}\frac{c}{((m-i)\lor 0)!}\abs*{\xi^2\paren*{\xi-1}^{(m-i)\lor 0}} \\
  \le& \sum_{m=1}^4\paren*{\alpha_{m-1}W+c_m}\sum_{i=0}^4\binom{4}{i}\frac{4c}{((m-i)\lor 0)!} \\
  \lesssim& 1.
\end{align}
For $\xi \in [\Delta,1]$, we have from \cref{lem:asm2lower} that
\begin{align}
  \abs*{\xi^2\phi^{(4)}(\xi)} \le \alpha_1W\xi^{\alpha-2} + c_4\xi^2 \lesssim \Delta^{\alpha-2}.
\end{align}
Since $\bar\phi_\Delta(\xi) = 0$ for $\xi \in [0,\Delta/2]$ and $\xi \ge 2$ by the construction, we have
\begin{align}
  \sup_{\xi \in \RealSet_+}\xi^2\abs*{\bar\phi_\Delta^{(4)}(\xi)} \lesssim \Delta^{\alpha-2}. \label{eq:reminder-xi}
\end{align}
Thus, we have
\begin{align}
  \abs*{\Mean\bracket*{R_3\paren*{\hat{p};\bar\phi_\Delta,p}}} \lesssim \frac{1}{n^2\Delta^{2-\alpha}} + \frac{1}{n^3\Delta^{3-\alpha}}. \label{eq:plugin-bias-r3}
\end{align}
Define $G(x) = \frac{1}{2}(\frac{\hat{p}}{x}-1)^2$. Then, the mean value theorem stats that there exists $\xi$ such that
\begin{align}
  R_1(\hat{p};\bar\phi^{(2)}_\Delta,p) =& \frac{\bar\phi_\Delta^{(4)}(\xi)}{2}(\hat{p}-\xi)\frac{\xi^2(\frac{\hat{p}}{p}-1)^2}{\hat{p}(\frac{\hat{p}}{\xi}-1)} \\
  =& \frac{\bar\phi_\Delta^{(4)}(\xi)}{2}\frac{\xi^3(\hat{p}-p)^2}{p^2\hat{p}}.
\end{align}
Thus, we have
\begin{align}
  \abs*{\Mean\bracket*{\frac{\hat{p}}{2n}R_1(\hat{p};\bar\phi^{(2)}_\Delta,p)}}
   \le& \Mean\bracket*{\frac{\abs*{\bar\phi_\Delta^{(4)}(\xi)}}{4n}\frac{\xi^3(\hat{p}-p)^2}{p^2}} \\
   \le& \frac{\sup_{\xi\in\RealSet_+}\xi^3\abs*{\bar\phi_\Delta^{(4)}(\xi)}}{4np^2}\Mean\bracket*{(\hat{p}-p)^2} \\
   =& \frac{1}{4n^2p}\sup_{\xi\in\RealSet_+}\xi^3\abs*{\bar\phi_\Delta^{(4)}(\xi)}.
\end{align}
In the similar manner of \cref{eq:reminder-xi}, we have
\begin{align}
  \sup_{\xi\in\RealSet_+}\xi^3\abs*{\bar\phi_\Delta^{(4)}(\xi)} \lesssim \Delta^{\alpha-1}.
\end{align}
Thus, we have
\begin{align}
  \abs*{\Mean\bracket*{\frac{\hat{p}}{2n}R_1(\hat{p};\bar\phi^{(2)}_\Delta,p)}}
   \lesssim \frac{1}{n^2p\Delta^{1-\alpha}} \le \frac{1}{n^2\Delta^{2-\alpha}}. \label{eq:plugin-bias-r4}
\end{align}
By the assumption $\Delta \gtrsim \frac{1}{n}$, we have $\frac{1}{n^3\Delta^{3-\alpha}} \lesssim \frac{1}{n^2\Delta^{2-\alpha}}$. Assembling \cref{eq:plugin-bias-r1,eq:plugin-bias-r2,eq:plugin-bias-r3,eq:plugin-bias-r4} gives the desired result.
\end{proof}
\begin{proof}[Proof of \cref{lem:plugin-var}]
From the property of the variance and the triangle inequality, we have
\begin{align}
  &\Var\bracket*{\bar\phi_\Delta\paren*{\frac{\tilde{N}}{n}} - \frac{\tilde{N}}{2n^2}\bar\phi_\Delta^{(2)}\paren*{\frac{\tilde{N}}{n}} - \phi(p) + \frac{p\phi^{(2)}(p)}{2n}} \\
  \le& \Mean\bracket*{\paren*{\bar\phi_\Delta\paren*{\frac{\tilde{N}}{n}} - \frac{\tilde{N}}{2n^2}\bar\phi_\Delta^{(2)}\paren*{\frac{\tilde{N}}{n}} - \phi(p) + \frac{p\phi^{(2)}(p)}{2n}}^2 } \\
  \le& 2\Mean\bracket*{\paren*{\bar\phi_\Delta\paren*{\frac{\tilde{N}}{n}} - \phi(p)}^2} + 2\Mean\bracket*{\paren*{\frac{\tilde{N}}{2n^2}\bar\phi_\Delta^{(2)}\paren*{\frac{\tilde{N}}{n}} - \frac{p\phi^{(2)}(p)}{2n}}^2 }. \label{eq:plugin-var-1}
\end{align}
Applying Taylor theorem to the first term of \cref{eq:plugin-var-1} gives
 \begin{align}
  &\abs*{\bar\phi_\Delta\paren*{\frac{\tilde{N}}{n}} - \phi(p)} \\
   =& \abs*{\phi^{(1)}(p)\paren*{\frac{\tilde{N}}{n} - p} + \frac{\phi^{(2)}(p)}{2}\paren*{\frac{\tilde{N}}{n} - p}^2 + \frac{\phi^{(3)}(p)}{6}\paren*{\frac{\tilde{N}}{n} - p}^3 + R_3\paren*{\frac{\tilde{N}}{n};\bar\phi_\Delta,p} },
 \end{align}
 where $R_3\paren*{\frac{\tilde{N}}{n};\bar\phi_\Delta,p}$ denotes the reminder term of the Taylor theorem. From the triangle inequality, we have
\begin{multline}
 \paren*{\bar\phi_\Delta\paren*{\frac{\tilde{N}}{n}} - \phi(p)}^2
  = 4\paren*{\phi^{(1)}(p)}^2\paren*{\frac{\tilde{N}}{n} - p}^2 + \paren*{\phi^{(2)}(p)}^2\paren*{\frac{\tilde{N}}{n} - p}^4 \\ + \frac{\paren*{\phi^{(3)}(p)}^2}{9}\paren*{\frac{\tilde{N}}{n} - p}^6 + 4\paren*{R_3\paren*{\frac{\tilde{N}}{n};\bar\phi_\Delta,p}}^2. \label{eq:plugin-var-2}
\end{multline}
The central moments for $X \sim \Poi(\lambda)$ are given as $\Mean[(X-\lambda)^2] = \lambda$,$\Mean[(X-\lambda)^4] = 3\lambda^2+\lambda$, and $\Mean[(X-\lambda)^6] = 15\lambda^3+25\lambda^2+\lambda$. \cref{lem:asm2lower}, the triangle inequality and the assumption $\frac{1}{n} \gtrsim \Delta$, the expectation of the first three terms in \cref{eq:plugin-var-2} have upper bounds as
\begin{align}
 \Mean\bracket*{4\paren*{\phi^{(1)}(p)}^2\paren*{\frac{\tilde{N}}{n} - p}^2} \le \frac{8W^2p^{2\alpha-1} + 8c_1^2p}{n}
  \lesssim \frac{p^{2\alpha-1}}{n} + \frac{p}{n},
\end{align}
\begin{align}
 \Mean\bracket*{\paren*{\phi^{(2)}(p)}^2\paren*{\frac{\tilde{N}}{n} - p}^4} \le& \paren*{2\alpha_1^2W^2p^{2\alpha-4} + c_2^2}\paren*{\frac{3p^2}{n^2} + \frac{p}{n^3}} \\
  \lesssim& \frac{p^{2\alpha-1}}{n^2\Delta} + \frac{p^{2\alpha-1}}{n^3\Delta^2} + \frac{p}{n^2} \\
  \lesssim& \frac{p^{2\alpha-1}}{n} + \frac{p}{n^2},
\end{align}
and
\begin{align}
 \Mean\bracket*{\frac{\paren*{\phi^{(3)}(p)}^2}{9}\paren*{\frac{\tilde{N}}{n} - p}^6} \le& \paren*{2\alpha_2^2W^2p^{2\alpha-6} + c_3^2}\paren*{\frac{15p^3}{n^3} + \frac{25p^2}{n^4} + \frac{p}{n^5}} \\
  \lesssim& \frac{p^{2\alpha-1}}{n^3\Delta^2} + \frac{p^{2\alpha-1}}{n^4\Delta^3} + \frac{p^{2\alpha-1}}{n^5\Delta^4} + \frac{p}{n^3} \\
  \lesssim& \frac{p^{2\alpha-1}}{n} + \frac{p}{n^3}.
\end{align}
From \cref{eq:r3-phi-reminder}, there exists $\xi$ between $p$ and $\hat{p}$ such that
\begin{align}
  & 4\Mean\bracket*{\paren*{R_3\paren*{\frac{\tilde{N}}{n};\bar\phi_\Delta,p}}^2} \\
  =& 4\Mean\bracket*{\paren*{\frac{\xi^3\bar\phi_\Delta^{(4)}(\xi)}{12p^2(\xi+\hat{p})}(\hat{p}-p)^4 }^2} \\
  \le& \frac{\sup_{\xi \in \RealSet_+}\abs*{\xi^2\bar\phi_\Delta^{(4)}(\xi)}^2}{36p^4}\Mean\bracket*{(\hat{p}-p)^8} \\
  \le& \paren*{\frac{105}{36n^4} + \frac{490}{36n^5p} + \frac{119}{36n^6p^2} + \frac{1}{36n^7p}}\sup_{\xi \in \RealSet_+}\abs*{\xi^2\bar\phi_\Delta^{(4)}(\xi)}^2,
\end{align}
\sloppy where we use $\Mean[(X-\lambda)^8]=105\lambda^4+490\lambda^3+119\lambda^2+\lambda$ for $X\sim\Poi(\lambda)$.
Since $\sup_{\xi \in \RealSet_+}\abs*{\xi^2\bar\phi_\Delta^{(4)}(\xi)}^2 \lesssim \Delta^{2\alpha-4}$ from \cref{eq:reminder-xi} and $\Delta\gtrsim\frac{1}{n}$ by the assumption, we have
\begin{align}
  & 4\Mean\bracket*{\paren*{R_3\paren*{\frac{\tilde{N}}{n};\bar\phi_\Delta,p}}^2} \\
  \lesssim& \frac{1}{n^4\Delta^{4-2\alpha}} + \frac{1}{n^5\Delta^{5-2\alpha}} + \frac{1}{n^6\Delta^{6-2\alpha}} + \frac{1}{n^7\Delta^{7-2\alpha}} \\
  \lesssim& \frac{1}{n^4\Delta^{4-2\alpha}}.
\end{align}
Letting $g(p)=p\bar\phi_\Delta^{(2)}(p)$, application of the Taylor theorem to the second term of \cref{eq:plugin-var-1} yields
\begin{align}
 \abs*{\frac{\tilde{N}}{2n^2}\bar\phi_\Delta^{(2)}\paren*{\frac{\tilde{N}}{n}} - \frac{p\phi^{(2)}(p)}{2n}}
  \le \frac{1}{2n}\abs*{\paren*{\phi{(2)}(p) + p\phi{(3)}(p)}\paren*{\frac{\tilde{N}}{n} - p} + R_1\paren*{\frac{\tilde{N}}{n};g,p}}.
\end{align}
The triangle inequality and $\Mean[(X-\lambda)^2]=\lambda$ for $X\sim\Poi(\lambda)$ give
\begin{align}
  &\Mean\bracket*{\paren*{\frac{\tilde{N}}{2n^2}\bar\phi_\Delta^{(2)}\paren*{\frac{\tilde{N}}{n}} - \frac{p\phi^{(2)}(p)}{2n}}^2} \\
  \le& \frac{(\phi{(2)}(p))^2 + (p\phi{(3)}(p))^2}{n^2}\Mean\bracket*{\paren*{\frac{\tilde{N}}{n} - p}^2} + \frac{1}{2n^2}\Mean\bracket*{\paren*{R_1\paren*{\frac{\tilde{N}}{n};g,p}}^2} \\
  =& \frac{p(\phi{(2)}(p))^2 + p(p\phi{(3)}(p))^2}{n^3} + \frac{1}{2n^2}\Mean\bracket*{\paren*{R_1\paren*{\frac{\tilde{N}}{n};g,p}}^2}.
\end{align}
Applying \cref{lem:asm2lower} gives
\begin{align}
  &\frac{p(\phi{(2)}(p))^2 + p(p\phi{(3)}(p))^2}{n^3} \\
  \le& \frac{1}{n^3}\paren*{2\alpha_1^2W^2p^{2\alpha-3} + 2pc_2^2 + 2\alpha_1^2W^2p^{2\alpha-3} + 2p^3c_3^2} \\
  \lesssim& \frac{p^{2\alpha-1}}{n^3\Delta^2} + \frac{p}{n^3} \\ \lesssim& \frac{p^{2\alpha-1}}{n} + \frac{p}{n^3}.
\end{align}
Let $\hat{p}=\frac{\tilde{N}}{n}$ and $G(x)=\frac{1}{x}(\hat{p}-x)^2$. Then, the mean value theorem gives that there exists $\xi$ between $p$ and $\hat{p}$ such that
\begin{align}
  \Mean\bracket*{\paren*{R_1\paren*{\hat{p};g,p}}^2}
   =& \Mean\bracket*{\paren*{g^{(1)}(\xi)(\hat{p}-\xi)\frac{G(\hat{p})-G(p)}{G^{(1)}(\xi)}}^2} \\
   =& \Mean\bracket*{\paren*{g^{(1)}(\xi)\frac{\xi^2(\hat{p}-p)^2}{p(\hat{p}+\xi)}}^2} \\
   \le& \paren*{\frac{3}{n^2} + \frac{1}{n^3p}}\sup_{\xi \in \RealSet_+}\abs*{\xi g^{(1)}(\xi)}^2 \\
   \le& \paren*{\frac{3}{n^2} + \frac{1}{n^3p}}\sup_{\xi \in \RealSet_+}\abs*{2\xi\bar\phi_\Delta^{(3)}(\xi) + \xi^2\bar\phi_\Delta^{(4)}(\xi)}^2\\
   \le& \paren*{\frac{3}{n^2} + \frac{1}{n^3p}}\paren*{2\sup_{\xi \in \RealSet_+}\abs*{\xi\bar\phi_\Delta^{(3)}(\xi)}^2 + 2\sup_{\xi \in \RealSet_+}\abs*{\xi^2\bar\phi_\Delta^{(4)}(\xi)}^2}.
\end{align}
In the similar manner of \cref{eq:reminder-xi}, we have
\begin{align}
  \sup_{\xi \in \RealSet_+}\abs*{\xi\bar\phi_\Delta^{(3)}(\xi)}^2 \lesssim \Delta^{2\alpha-4}, \textand \sup_{\xi \in \RealSet_+}\abs*{\xi^2\bar\phi_\Delta^{(4)}(\xi)}^2 \lesssim \Delta^{2\alpha-4}.
\end{align}
Thus, we have
\begin{align}
  \frac{1}{2n^2}\Mean\bracket*{\paren*{R_1\paren*{\frac{\tilde{N}}{n};g,p}}^2} \lesssim \frac{1}{n^4\Delta^{4-2\alpha}} + \frac{1}{n^5\Delta^{5-2\alpha}} \lesssim \frac{1}{n^4\Delta^{4-2\alpha}}.
\end{align}
Consequently, we get the bound of the variance as
\begin{align}
  \frac{p^{2\alpha-1}}{n} + \frac{1}{n^4\Delta^{4-2\alpha}} + \frac{p}{n}.
\end{align}
\end{proof}

\section{Proof of \cref{prop:no-estimator}}\label{sec:proof-prop-no-estimator}
\begin{proof}[Proof of \cref{prop:no-estimator}]
  It is obviously that if the output domain of $\phi$ is unbounded, i.e., there is a point $p_0 \in [0,1]$ such that $\abs*{\phi(p)} \to \infty$ as $p\to p_0$, there is no consistent estimator. Letting $p_0 = \paren*{\frac{W}{W\lor -c'_1}}$, $\phi^{(1)}(p)$ has same sign in $(0,p_0]$. Thus, for any $p \in (0,p_0]$, we have
  \begin{align}
    \abs*{\phi(p) - \phi(p_0)} =& \abs*{\int_{p_0}^p \phi^{(1)}(x) dx} \\
    =& \int_p^{p_0} \abs*{\phi^{(1)}(x)} dx \\
    \ge& W \int_p^{p_0} p^{-1} dx + c'_1(p_0 - p) \\
    \ge& W \ln(p_0/p) + c'_1(p_0-p).
  \end{align}
  Since $\abs*{\phi(p) - \phi(p_0)} \to \infty$ as $p \to 0$, $\phi$ is unbounded and we gets the claim.
\end{proof}

\section{Additional Lemmas}
Here, we introduce some additional lemmas and their proofs.
\begin{lemma}\label{lem:asm2lower}
  For a non-integer $\alpha$, let $\phi$ be a $m$ times continuously differentiable function on $(0,1]$ where $m \ge 1 + \alpha$. Suppose that there exist finite constants $W > 0$, $c_m$ and $c'_m$ such that
  \begin{align}
    \abs*{\phi^{(m)}(p)} \le \alpha_{m-1} Wp^{\alpha-m} + c_m, \textand \abs*{\phi^{(m)}(p)} \ge \alpha_{m-1} Wp^{\alpha-m} + c'_m.
  \end{align}
  Then, there exists finite constants $c_{m-1}$ and $c'_{m-1}$ such that
  \begin{align}
    \abs*{\phi^{(m-1)}(p)} \le \alpha_{m-2} Wp^{\alpha-m+1} + c_{m-1}, \textand \abs*{\phi^{(m-1)}(p)} \ge \alpha_{m-2} Wp^{\alpha-m+1} + c'_{m-1},
  \end{align}
  where $\alpha_{0} = 1$ and $\alpha_i = \prod_{j=1}^i(j-\alpha)$ for $i = 1,...,m$.
\end{lemma}
\begin{proof}[Proof of \cref{lem:asm2lower}]
 Let $p_m = \paren*{\frac{\alpha_{m-1}W}{\alpha_{m-1}W \lor -c'_m}}^{1/(m-\alpha)}$. Then, $\abs*{\phi^{(m)}(p)} > 0$ for $p \in (0,p_m)$. From continuousness of $\phi^{(m)}$, $\phi^{(m)}(p)$ has same sign in $p \in (0,p_m]$, and thus we have either $\phi^{(m)}(p) \ge \alpha_{m-1}Wp^{\alpha-m} + c'_m$ or $\phi^{(m)}(p) \le -\alpha_{m-1}Wp^{\alpha-m} - c'_m$ in $p \in (0,p_m]$.  Since $\phi^{(m-1)}$ is absolutely continuous on $(0,1]$, we have for any $p \in (0,1]$
 \begin{align}
   \phi^{(m-1)}(p) = \phi^{(m-1)}(p_m) + \int_{p_m}^p \phi^{(m)}(x) dx.
 \end{align}
 The absolute value of the second term has an upper bound as
 \begin{align}
  \abs*{\int_{p_m}^p \phi^{(m)}(x) dx} \le& \abs*{\int_{p_m}^p \alpha_{m-1}Wx^{\alpha-m} + c_m dx} \\
   \le& \abs*{\alpha_{m-2} W \paren*{p_m^{\alpha-m+1}-p^{\alpha-m+1}} + c_m\paren*{p - p_m}} \\
   \le& \alpha_{m-2} W p^{\alpha-m+1} + \abs*{\alpha_{m-2} W p_m^{\alpha-m+1} + c_m(p_m-p)} \\
   \le& \alpha_{m-2} W p^{\alpha-m+1} + \alpha_{m-2} W p_m^{\alpha-m+1} + \abs*{c_m}.
 \end{align}
 Also, we have a lower bound of the second term as
 \begin{align}
  \abs*{\int_{p_m}^p \phi^{(m)}(x) dx}
   =& \abs*{\int_{p_m}^{p \land p_m} \phi^{(m)}(x) dx + \int_{p \land p_m}^p \phi^{(m)}(x) dx}\\
   \ge& \abs*{\int_{p_m}^{p \land p_m} \alpha_{m-1}Wx^{\alpha-m} + c'_m dx} - \abs*{\int_{p \land p_m}^{p} \alpha_{m-1}W p_m^{\alpha-m} + c_m dx} \\
   \ge& \begin{multlined}[t]
    \abs*{\alpha_{m-2} W\paren*{p_m^{\alpha-m+1} - \paren*{p \land p_m}^{\alpha-m+1}} +  c'_m((p \land p_m) - p_m)} \\ - \abs*{\paren*{\alpha_{m-1}W p_m^{\alpha-m} + c_m}(p - (p \land p_m))}
   \end{multlined} \\
   \ge& \begin{multlined}[t]
    \alpha_{m-2} W\paren*{p \land p_m}^{\alpha-m+1} - \alpha_{m-2} Wp_m^{\alpha-m+1} - \abs*{c'_m(p_m - (p \land p_m))} \\ - \paren*{\alpha_{m-1}W p_m^{\alpha-m} + c_m}(p - (p \land p_m))
   \end{multlined} \\
   \ge& \begin{multlined}[t]
    \alpha_{m-2} Wp^{\alpha-m+1} -  \alpha_{m-2} Wp_m^{\alpha-m+1} - \abs*{c'_m}p_m \\ - \paren*{\alpha_{m-1} W p_m^{\alpha-m} + c_m}(1 - p_m)
   \end{multlined}.
 \end{align}
 Applying the triangle inequality and the reverse triangle inequality gives
 \begin{align}
    \abs*{\int_{p_m}^p \phi^{(m)}(x) dx} - \abs*{\phi^{(m-1)}(p_m)}\le \abs*{\phi^{(m-1)}(p)} \le \abs*{\int_{p_m}^p \phi^{(m)}(x) dx} + \abs*{\phi^{(m-1)}(p_m)}.
 \end{align}
 \sloppy Thus, setting $c_{m-1} = \alpha_{m-2} W p_m^{\alpha-m+1} + \abs*{c_m} + \abs*{\phi^{(m-1)}(p_m)}$ and $c'_{m-1} = -\alpha_{m-2} Wp_m^{\alpha-m+1} - \abs*{c'_m}p_m - \paren*{\alpha_{m-1} W p_m^{\alpha-m} + c_m}(1 - p_m) - \abs*{\phi^{(m-1)}(p_m)}$ yields the claim.
\end{proof}

\begin{lemma}\label{lem:est-bounded}
  Under \cref{asm:bound} or \cref{asm:bound2}, for any $p, p' \in [0,1]$
  \begin{align}
    \abs*{\phi(p) - \phi(p')} \le \frac{W}{\alpha}\abs*{p-p'}^{\alpha} + \abs*{c_1(p-p')}.
  \end{align}
\end{lemma}
\begin{proof}[Proof of \cref{lem:est-bounded}]
  We can assume $p' \le p$ without loss of generality. The absolute continuously of $\phi$ gives
  \begin{align}
    \abs*{\phi(p) - \phi(p')} = \abs*{\int_{p'}^p \phi^{(1)}(x) dx} \le \abs*{\int_{p'}^p \abs*{\phi^{(1)}(x)} dx}.
  \end{align}
  From \cref{lem:asm2lower}, we have
  \begin{align}
    \abs*{\phi(p) - \phi(p')} \le& \abs*{\int_{p'}^p (Wx^{\alpha-1} + c_1) dx} \\
     =& \abs*{\frac{W}{\alpha}\paren*{p^{\alpha} - p'^{\alpha}} + c_1(p-p')} \\
     \le& \frac{W}{\alpha}\abs*{p-p'}^{\alpha} + \abs*{c_1(p-p')},
  \end{align}
  where the last line is obtained since a function $x^{\alpha}$ for $\alpha \in (0,1)$ is $\alpha$-Holder continuous. This is valid for the case $p' = 0$. Indeed,
  \begin{align}
    \abs*{\phi(p) - \phi(0)} =& \lim_{p' \to 0}\abs*{\phi(p) - \phi(p')} \\
     \le& \lim_{p' \to 0} \paren*{\frac{W}{\alpha}\abs*{p-p'}^{\alpha} + \abs*{c_1(p-p')}} \\
     =& \frac{W}{\alpha}\abs*{p-0}^{\alpha} + \abs*{c_1(p-0)}.
  \end{align}
\end{proof}

\begin{lemma}\label{lem:bound-sum-alpha}
  Given $\alpha \in [0,1]$, $\sup_{P \in \dom{M}_k} \sum_{i=1}^k p_i^\alpha = k^{1-\alpha}$.
\end{lemma}
\begin{proof}[Proof of \cref{lem:bound-sum-alpha}]
  If $\alpha = 1$, the claim is obviously true. Thus, we assume $\alpha < 1$. We introduce the Lagrange multiplier $\lambda$ for a constraint $\sum_{i=1}^n p_i = 1$, and let the partial derivative of $\sum_{i=1}^k p_i^\alpha + \lambda(1 - \sum_{i=1}^k p_i)$ with respect to $p_i$ be zero. Then, we have
  \begin{align}
   \alpha p_i^{\alpha-1} - \lambda = 0. \label{eq:sum-alpha-diff}
  \end{align}
  Since $p^{\alpha-1}$ is a monotone function, the solution of \cref{eq:sum-alpha-diff} is given as $p_i = (\lambda/\alpha)^{1/(\alpha-1)}$, i.e., the values of $p_1,...,p_k$ are same. Thus, the function $\sum_{i=1}^k p_i^\alpha$ is maximized at $p_i = 1/k$ for $i=1,...,k$. Substituting $p_i=1/k$ into $\sum_{i=1}^k p_i^\alpha$ gives the claim.
\end{proof}

\begin{lemma}\label{lem:bound-sum-alpha-minus}
  Given $\alpha < 0$ and $\Delta \le \frac{1}{k}$, $\sup_{P \in \dom{M}_k : \forall i, p_i \ge \Delta} \sum_{i=1}^k p_i^\alpha = \paren*{(1-(k-1)\Delta)^\alpha + (k-1)\Delta^\alpha} \le k\Delta^\alpha$.
\end{lemma}
\begin{proof}
  From the Karush–Kuhn–Tucker conditions, letting $P^*=(p^*_1,...,p^*_k)$ be a probability vector that attains the supremum, there exist real values $\lambda$ and $\delta_i \ge 0$ such that
  \begin{align}
    (p^*_i)^{\alpha-1} - \lambda - \delta_i = 0,
  \end{align}
  and $p^*_i = \Delta$ only if $\delta_i > 0$. Thus, we have
  \begin{align}
    p^*_i = \lambda^{1/(\alpha-1)} \textor p^*_i = \Delta.
  \end{align}
  Hence,
  \begin{align}
    \sup_{P \in \dom{M}_k : \forall i, p_i \ge \Delta} \sum_{i=1}^k p_i^\alpha = \max_{m = 1,...,k-1}\paren*{m\Delta^\alpha + (k-m)(1-m\Delta)^\alpha}.
  \end{align}
  Since $\Delta^\alpha \ge (1-m\Delta)^\alpha$ for $m = 1,...,k-1$, the maximum is attained at $m = k-1$. Moreover, we have $(1-(k-1)\Delta)^\alpha \le \Delta^\alpha$, and thus we get the claim.
\end{proof}

\end{document}